\documentclass[12pt]{article}
\usepackage[english]{babel}
\usepackage{graphicx}
\usepackage{amsmath}
\usepackage{amssymb}
\usepackage{amsthm}
\usepackage{natbib}
\usepackage{enumerate}
\usepackage{float}
\usepackage{bbm}
\usepackage{booktabs}
\usepackage{subcaption}
\usepackage{textcomp}
\usepackage{mathtools}
\usepackage{tikz}
\usepackage{caption}
\usepackage{multirow}
\usepackage{cancel}

\allowdisplaybreaks[2]

\usepackage{color}
\definecolor{DarkBlue}{rgb}{0,0.18,0.55}
\usepackage{hyperref}
\hypersetup{pdfauthor={Kim, Kim, and Pal},colorlinks=true,citecolor=DarkBlue,filecolor=DarkBlue,linkcolor=DarkBlue,urlcolor=DarkBlue,pdftex}
\usepackage[flushleft]{threeparttable}
\usepackage{xargs}
\usepackage{array}

\theoremstyle{plain}
\newtheorem{theorem}{Theorem}
\newtheorem{lemma}{Lemma}
\newtheorem{prop}{Proposition}

\newtheorem{assumption}{Assumption}
\newtheorem{observation}{Observation}
\newtheorem{condition}{Condition}
\newtheorem{remark}{Remark}

\def\ci{\perp\!\!\!\perp}

\begin{document}
	{\def\thanks#1{\protect\thanksaliased{#1}}%
		\let\thanksaliased\thanks
		\title{{Partial Identification of the Valuation Distribution\\ in Sequential English Auctions}\thanks{We thank participants at MEG 2022, CEA 2025, and the 31st SJE International Symposium: Partial Identification in Econometrics at Seoul National University for valuable comments. Dongwoo Kim gratefully acknowledges support from the Social Sciences and Humanities Research Council of Canada under the Insight Grant (435-2024-0322)}}
		\author{Dongwoo Kim\thanks{Department of Economics, Simon Fraser University and Korea University. Email: \href{mailto:dongwook@sfu.ca}{\tt dongwook@sfu.ca}} \and Kyoo il Kim\thanks{Department of Economics, Michigan State University. Email: \href{mailto:kyookim@msu.edu}{\tt kyookim@msu.edu}} \and Pallavi Pal\thanks{School of Business, Stevens Institute of Technology. Email: \href{mailto:ppal2@stevens.edu}{\tt ppal2@stevens.edu}}}
	}%
	\date{}
	\maketitle \pagenumbering{gobble}
	
	\vspace{-6ex}
	\begin{abstract}\noindent
		This paper extends the incomplete model of \cite{haile2003inference} from static English auctions to sequential English auctions. Because bidders may wait for future opportunities, the static condition that bidders do not let rivals win at beatable prices need not hold. We replace it with a dynamic opportunity-cost restriction, yielding nonparametric valuation bounds without solving a dynamic equilibrium. Sharp bounds are also characterized. We propose a novel moment-condition inversion estimator that pools auctions with heterogeneous bidder counts, mitigating finite-sample instability of order statistics approaches and admitting analytical standard errors and smooth confidence intervals. Applications to Korean wholesale used-car auctions and Cars and Bids online auctions deliver informative bounds. Counterfactual analyses show that the option to wait lowers first-period revenue by 8--11\% in the Korean market, that increasing effective competition from 8 to 20 serious bidders in Cars and Bids raises seller revenue by 40--65\%, and that maximin reserve prices vary substantially across vehicle clusters.
		
		\vspace{1ex}
		
		
	\end{abstract}
	
	\newpage\pagenumbering{arabic}
	
	\section{Introduction}\label{sec:intro}
	
	Sequential auctions are common in markets for vintage wines, paintings, used cars, and on online platforms such as eBay. Bidders in these markets decide not only how much they are willing to pay for the current object, but also whether to preserve the option to compete for later objects. This trade-off makes observed bids difficult to interpret. A bidder who stops early may have a low valuation, or may value the current item highly but prefer to wait for a future opportunity. The theoretical literature has formalized this dynamic trade-off in many ways, with different models generating different bidding predictions. Risk aversion, declining marginal values, supply uncertainty, loss aversion, and other mechanisms can each rationalize different price paths and bidding strategies. This multiplicity of models, combined with the well-documented difficulty of computing equilibria in sequential auctions, raises a fundamental question for empirical work: can we learn about the valuation distribution without committing to a specific equilibrium?
		
	In this paper, we develop a partial identification framework for sequential English auctions, a prevalent format for sequential sales, that bounds the distribution of bidders' valuations using only two behavioral assumptions: (i) bidders do not bid above their valuations, and (ii) bidders stop bidding when the opportunity cost of winning the current auction exceeds the current-period profit. These assumptions define an \emph{incomplete model} in the sense of \cite{haile2003inference} (hereafter HT): they restrict bidder behavior without specifying which equilibrium is played. These assumptions are substantially weaker than those required for a full dynamic equilibrium and are compatible with a broad class of sequential-auction models. The resulting bounds are therefore robust to many forms of model uncertainty that complicate structural estimation.
		
	Knowledge of the valuation distribution (denoted by $F$ hereafter) is central to auction design.  It helps determine reserve prices and expected revenue under alternative formats \citep{tang2011bounds, kim2025searchauction}, as well as the design and ordering of heterogeneous items in sequential sales \citep{elmaghraby2003importance, muramoto2016sequential, shi2022implementing}.  In sequential auction markets, where equilibrium-based estimates of~$F$ may be misspecified, informative \emph{bounds} on~$F$ are practically valuable: they characterize the competitiveness of the market, inform the design of auction mechanisms, and serve as inputs to counterfactual policy analysis that is robust to equilibrium assumptions.
	
	The paper's contributions fall into three groups. First, on identification, we extend the HT approach to sequential English auctions of heterogeneous objects. HT's static condition no longer holds in a sequential setting, since a bidder may decline to outbid a rival because the future option is more valuable than winning now. We formalize this dynamic opportunity cost and derive bounds on valuations without specifying dynamic equilibrium bidding strategies, while accommodating heterogeneous objects through the common value component. The resulting bounds on the valuation distribution are derived using an order-statistic inversion approach. When individual bidders can be tracked across auctions, the panel structure provides additional identifying information. A bidder's maximum reduced bid across all periods in which she participates is a tighter lower bound on her private value than any single-period bid, yielding a direct bound on $F$ without order-statistic inversion. We also characterize sharp bounds on $F$ using the generalized instrumental variable (GIV) framework \citep{chesher2017generalized}.\footnote{It is known that HT bounds in the static case are not sharp; this extends to our setting. The GIV framework provides a sharp characterization based on random set theory without requiring constructive proof. We show that the non-sharp bounds from our main analysis are nested as special cases and identify the additional moment inequalities that tighten the identified set.} 
	
	Second, on estimation and inference, we develop a \emph{moment-condition inversion} estimator that pools auctions with heterogeneous numbers of bidders into a single estimation step.  In the original HT application, bounds are estimated by conditioning on~$N$ via kernel smoothing, which can downweight much of the data and can produce crossing bounds when the bandwidth is too wide.  More broadly, order-statistic inversion is poorly behaved in finite samples near the support boundaries, as emphasized by \citet{menzel2013large}. Our approach instead solves a moment equation, where each auction contributes with a weight determined by its own number of bidders: low-competition auctions carry information in the left tail, while high-competition auctions contribute more in the right tail.  This mitigates the tail degeneracy of fixed-$N$ order-statistic inversion by exploiting variation in bidder counts across auctions. We establish consistency, derive the asymptotic distribution, and use analytical standard errors to construct confidence intervals, with bootstrap intervals serving as a finite-sample benchmark.  To our knowledge, this pooling approach has not been used previously in the auction literature.
	
	Third, we demonstrate the practical value of the approach through simulation experiments, two empirical applications, and three counterfactual exercises.  The simulations show that applying the static HT bounds to sequential auction data produces misspecified (crossing) bounds, while our method yields informative bounds that are valid under the maintained assumptions and performs reliably in finite samples. The two applications illustrate complementary strengths of the sequential approach: in Korean wholesale used-car auctions, where terminal auctions and long within-day sequences are observed, the sequential and terminal-period inequalities deliver tight bounds on $F$ across vehicle categories. In Cars and Bids online auctions, where listings recur continuously with no terminal period and the static HT lower bound is inapplicable, the sequential approach is the only available lower bound. Our counterfactual analyses translate the bounds on~$F$ into bounds on policy quantities: the option to wait lowers first-period revenue by 8--11\% in the Korean market, increasing effective competition from 8 to 20 serious bidders in Cars and Bids raises seller revenue by 40--65\%, and the maximin reserve varies substantially across vehicle clusters.
	
	Our paper connects three literatures.  In sequential-auction theory and closely related auction models, the canonical predictions are a martingale price process under IPV \citep{weber1983multiobject} and rising prices under affiliation \citep{milgrom1982theory}, yet the ``declining price anomaly'' is pervasive.\footnote{This anomaly is well documented in wine \citep{ashenfelter1989auctions, ashenfelter2003auctions}, art \citep{beggs1997declining}, and flowers \citep{vandenberg2001winner}, though the direction is not universal. \citet{raviv2006new} documents rising prices in New Jersey used car auctions.}  Competing explanations include risk aversion \citep{mcafee1993declining}, declining marginal values \citep{bernhardt1994note, kittsteiner2004declining}, buyer's options \citep{black1992winner}, supply uncertainty \citep{jeitschko1999equilibrium, engelbrecht1994sequential}, synergies \citep{branco1997sequential, kong2021sequential}, loss aversion \citep{rosato2022loss}, ambiguity \citep{ghosh2021sequential}, stochastic entry \citep{hendricks2012last, said2011sequential}, and item ordering effects \citep{elmaghraby2003importance, muramoto2016sequential, shi2022implementing}.  Equilibrium computation in this setting is technically difficult \citep{caillaud2004equilibrium, benoit2001multiple, landi2018sequential}, and the bid function can change qualitatively under small perturbations.  Our behavioral restrictions are compatible with many of these mechanisms because they do not require a particular equilibrium bid function.
	
	For nonparametric estimation of auctions, the dominant approach is structural: specify an equilibrium and invert the bid function \citep{guerre2000optimal, athey2002identification, athey2007nonparametric, kim2025searchauction}.  Extensions address affiliated values \citep{li2002structural}, unknown numbers of bidders \citep{il2014nonparametric, an2010estimating}, unobserved heterogeneity \citep{krasnokutskaya2011identification}, and selective entry \citep{gentry2014partial}.  For dynamic auctions, structural methods require solving the full dynamic programming problem \citep{jofre2003estimation, donald2006empirical, kong2021sequential, groeger2014study}.  For sequential English auctions specifically, \citet{brendstrup2006identification} achieve point identification by exploiting changing bidder composition, and \citet{brendstrup2007non} proposes a nonparametric estimator, but \citet{lamy2010identification} shows the latter's identification argument fails.  This identification failure further motivates partial identification, which does not require invertibility of an equilibrium bid function. In partial identification of auctions, \citet{haile2003inference} propose the incomplete model for static English auctions.  \citet{tang2011bounds} demonstrates policy relevance by computing bounds on counterfactual revenue.  Our paper extends HT to the sequential setting; to the best of our knowledge, it is the first to provide partial identification bounds for valuations in sequential English auctions.
	
	The remainder of the paper is organized as follows.  Section~\ref{sec:model} presents the model. Section~\ref{sec:identification} derives bounds on the valuation distribution, characterizes sharp bounds, and discusses robustness. Section~\ref{sec:estimation_inference} discusses estimation and inference. Section~\ref{sec:sim} presents simulations. Section~\ref{sec:empirical} applies the method to two datasets, Korean wholesale used car auctions and online vehicle auctions from Cars and Bids.  Section~\ref{sec:counterfactual} presents three counterfactual exercises that translate the bounds on~$F$ into bounds on policy quantities (future-auction uncertainty, effective competition, and reserve-price design). Section~\ref{sec:conclusion} concludes. All proofs are in the appendix.
		
	\section{The Incomplete Sequential Auction Model}\label{sec:model}
	
	An auctioneer holds a finite series of ascending-price auctions with a single heterogeneous unit sold in each period $k \in \{1, \ldots, K\}$. Bidders, indexed by $i \in \{1, \ldots, N\}$, enter the market before the start of the sequential auctions and remain until they win an item or the final auction concludes. Each bidder is risk-neutral with single-unit demand. Supply uncertainty is captured by $\tau_{i,k} \in [0,1]$, the probability that an item is available for sale in period $k$ from bidder $i$'s perspective.
	
	The items are heterogeneous and imperfect substitutes. Bidder $i$'s valuation for the object in period $k$ is $v_{i,k} \equiv \phi(Z_k)\, \theta_i$, where $\phi(Z_k)$ is a common value component depending on observable item characteristics $Z_k \in \mathcal{R}_Z$ with $\phi: \mathcal{R}_Z \to \mathbb{R}_{+}$, and $\theta_i$ is bidder $i$'s private value component, drawn i.i.d.\ from a distribution $F^0(\cdot)$ supported on $[\underline{\theta}, \bar{\theta}]$ and independent of $Z_k$. Bidders know their own $\theta_i$ and $\phi(Z_k)$ but not the private values of their competitors. The total number of bidders $N$ and the number of auctions $K$ are common knowledge. Each auction follows an open ascending-price format with a fixed bid increment $\Delta \geq 0$. Bidding starts at an opening price $p_k$ set at or near the reserve price $r_k$. The auctioneer accepts monotonically increasing bids, and the auction ends when only one active bidder remains; the winner pays the final bid. Let $b_{i,k}$ denote bidder $i$'s highest bid in period $k$. We impose the same multiplicative structure on bids:
	\begin{align}\label{eq:bid_structure}
		b_{i,k} \equiv \phi(Z_k)\, \eta_{i,k}, \quad \text{where } \ \eta_{i,k} \sim G_k \ \text{and } \eta_{i,k} \ci Z_k.
	\end{align}
	This structure provides the basis for recovering the common value component $\phi(Z_k)$ from observed bids and item characteristics.
	
	Our incomplete model consists of two behavioral assumptions.
	
	\begin{assumption}\label{ass:no_overbid}
		Bidders do not bid above their valuations: $b_{i,k} \leq v_{i,k}$ for all $i$ and $k$.
	\end{assumption}
	
	\begin{assumption}\label{ass:opportunity_cost}
		In each period, a bidder who does not win the current auction has a weakly higher expected payoff from future periods than the profit from winning at the current price.
	\end{assumption}
	
	\noindent Assumption~\ref{ass:no_overbid} is standard and identical to the first assumption in HT: no rational bidder accepts a sure loss. Assumption~\ref{ass:opportunity_cost} is the dynamic analogue of HT's second assumption: losing bidders weakly prefer their continuation value to the profit from winning at the current price. We impose no further restrictions on the bidding strategy; our model permits any behavior consistent with these two assumptions.
	
	The expected payoff from future periods after period $k$ for bidder $i$ is
	\begin{multline}\label{eq:W}
		W_k(\theta_i) \equiv \sum_{t=k+1}^{K} \tau_{i,t} \underbrace{\left(\prod_{r=k+1}^{t-1}\big[1 - P(b_{i,r} > \max_{j \neq i} b_{j,r})\big]\right)}_{\text{prob.\ of losing in periods $k\!+\!1, \ldots, t\!-\!1$}} \\
		\times \underbrace{P(b_{i,t} > \max_{j \neq i} b_{j,t})\, \mathbf{E}\!\left[\pi_{i,t} \mid b_{i,t} > \max_{j \neq i} b_{j,t}\right]}_{\text{expected profit in period $t$}},
	\end{multline}
	where $\pi_{i,t}$ is bidder $i$'s profit from winning in period $t$.\footnote{The product $\prod_{r=k+1}^{t-1}$ equals $1$ when $t=k+1$ (an empty product).} This expected future payoff represents the opportunity cost of winning in the current period.\footnote{We define the expected payoff as the anticipated future payoff after bidder $i$ has lost in period $k$. Under Assumption~\ref{ass:opportunity_cost}, the bidder weighs the profit from winning the current auction against this opportunity cost.} Importantly, we do not impose any structure on the bidder's beliefs about future prices or winning probabilities.
	
	We are primarily interested in the primitive distribution $F^0(\cdot)$. However, if the opening price $p_k > \phi(Z_k)\underline{\theta}$ is binding in period $k$, the auction is uninformative about $F^0$ below $p_k/\phi(Z_k)$. Let $\tilde{p}_k \equiv p_k/\phi(Z_k)$ denote the opening price adjusted by the common value component. Conditional on $\tilde{p}_k$, the truncated distribution is
	\begin{align}\label{eq:truncated_conditional}
		F(\theta \mid \tilde{p}_k) = \frac{F^0(\theta) - F^0(\tilde{p}_k)}{1 - F^0(\tilde{p}_k)} \cdot \mathbbm{1}[\theta \geq \tilde{p}_k], \quad \forall\, \theta \in [\underline{\theta}, \bar{\theta}].
	\end{align}
	Values of participating bidders are i.i.d.\ draws from this truncated distribution, and the unconditional truncated distribution is
	\begin{align}\label{eq:truncated_unconditional}
		F(\theta) = \mathbf{E}_{\tilde{p}}\!\left[F(\theta \mid \tilde{p})\right].
	\end{align}
	Following HT, we treat $F(\cdot)$ as the target distribution of interest, because observed bids identify only the truncated distribution. As discussed in HT, $F(\cdot)$ is sufficient for answering normative and positive questions such as solving for optimal reserve prices.
	
	\section{Identification of the Valuation Distribution}\label{sec:identification}
	
	We first derive bounds on bidder valuations using Assumptions~\ref{ass:no_overbid}--\ref{ass:opportunity_cost}, then construct bounds on the valuation distribution. Throughout, $n_k$ denotes the number of participating bidders in period $k$ and $I_k$ their index set. We write $\theta^{j:n_k} \sim F_{j:n_k}(\cdot)$ for the $j$th highest order statistic of private values in period $k$, $b^{j:n_k}$ for the corresponding order statistic of bids, and $\eta^{j:n_k} \equiv b^{j:n_k}/\phi(Z_k) \sim G_{j:n_k}(\cdot)$ for the reduced-bid order statistic implied by \eqref{eq:bid_structure}.
	
	\subsection{Bounds on Bidder Valuations}\label{subsec:continuation_filter}
	
	We begin with the illustrative two-period case before stating the general results. Consider $K = 2$ and set $p_1=p_2=0$ for simplicity. Under Assumption~\ref{ass:no_overbid}, bidder $i$'s private component is bounded below by her reduced bid in every period. In the last period, there is no future auction, so Assumption~\ref{ass:opportunity_cost} implies the same result as in the static case: all participating bidders who bid below the highest bid have valuations below $(b^{1:n_2} + \Delta)/\phi(Z_2)$. In the first period, Assumption~\ref{ass:opportunity_cost} provides an additional upper bound. All bidders who bid below the highest bid have a weakly higher expected payoff from period 2 than the profit from winning in period 1:
	\begin{align}\label{eq:two_period_ub}
		\theta_i \phi(Z_1) - (b^{1:n_1} + \Delta) &\leq \tau_{i,2}\, P(b_{i,2} > \max_{j \neq i} b_{j,2})\, \mathbf{E}[\pi_{i,2} \mid b_{i,2} > \max_{j \neq i} b_{j,2}] \nonumber \\
		&\leq \mathbf{E}[\pi_{i,2} \mid b_{i,2} > \max_{j \neq i} b_{j,2}] \nonumber \\
		&\leq \theta_i \phi(Z_2).
	\end{align}
	The left-hand side is the profit from winning the first auction by raising the bid slightly above the highest bid. The second line uses $P(\cdot) \leq 1$ and $\tau_{i,2} \leq 1$, and the last line bounds the expected profit by the maximum possible gross value of the future object. Suppose that $\phi(Z_1) > \phi(Z_2)$. Rearranging yields
	\begin{equation}\label{eq:two_period_bound}
		\theta_i \leq \frac{b^{1:n_1} + \Delta}{\phi(Z_1) - \phi(Z_2)}, \quad \forall\, i \in I_1 \text{ with } b_{i,1} < b^{1:n_1}.
	\end{equation}
	
	Now turning to the general $K$-period case, Assumption~\ref{ass:no_overbid} implies that valuations are bounded below by observed bids:
	\begin{lemma}[Lower bound on valuations]\label{lem:lb}
		Under Assumption~\ref{ass:no_overbid},
		\begin{align}\label{eq:value_lb}
			\eta^{j:n_k} = \frac{b^{j:n_k}}{\phi(Z_k)} \leq \theta^{j:n_k}, \quad \forall\, j \in \{1, \ldots, n_k\} \text{ and } k \in \{1, \ldots, K\}.
		\end{align}
	\end{lemma}
	
	\noindent In the last period, Assumption~\ref{ass:opportunity_cost} reduces to HT's second assumption, since there is no future auction. Hence we derive the following upper bound:
	\begin{lemma}[Upper bound---last period]\label{lem:ub_last}
		Under Assumption~\ref{ass:opportunity_cost}, for all $i \in I_K$,
		\begin{equation}\label{eq:ub_last}
			\theta_i \leq \begin{cases}
				\bar{\theta}, & \text{if } b_{i,K} = b^{1:n_K}, \\[4pt]
				\displaystyle\frac{b^{1:n_K} + \Delta}{\phi(Z_K)}, & \text{if } b_{i,K} < b^{1:n_K}.
			\end{cases}
		\end{equation}
	\end{lemma}
	\noindent For a non-terminal period $k < K$, Assumption~\ref{ass:opportunity_cost} implies
	\begin{align}\label{eq:opportunity_cost_ineq}
		\theta_i \phi(Z_k) - (b^{1:n_k} + \Delta) \leq W_k(\theta_i), \quad \forall\, i \in I_k \text{ with } b_{i,k} < b^{1:n_k},\  k < K.
	\end{align}
	The key step is bounding the expected future payoff $W_k(\theta_i)$. For a retained auction $k$, we specify a continuation-value index $C_k$ and a nonnegative payment term $m_k$ such that
	\begin{equation}\label{eq:Ck_filter}
		W_k(\theta_i) \leq \theta_i C_k - m_k
	\end{equation}
	for the non-winning bidders. The index $C_k$ is an upper bound on the common-value component of the relevant future opportunity set, and $m_k$ is a lower bound on the payment required to realize that opportunity.  Together, they give an observable upper bound on the net continuation payoff.  Combining \eqref{eq:Ck_filter} with \eqref{eq:opportunity_cost_ineq} gives
	\begin{equation}\label{eq:Ck_prebound}
		\theta_i\{\phi(Z_k)-C_k\} \leq b^{1:n_k}+\Delta-m_k.
	\end{equation}
	Thus any auction satisfying $\phi(Z_k)>C_k$ yields an upper bound on non-winner valuations,
	\begin{equation}\label{eq:Ck_upper}
		\theta_i \leq \frac{b^{1:n_k}+\Delta-m_k}{\phi(Z_k)-C_k}, \quad i \in I_k,\ b_{i,k}<b^{1:n_k},
	\end{equation}
	whenever the numerator is positive.  Conversely, any auction satisfying $\phi(Z_k)<C_k$ yields a lower bound on valuations,
	\begin{equation}\label{eq:Ck_lower}
		\theta_i \geq \frac{m_k-b^{1:n_k}-\Delta}{C_k-\phi(Z_k)}, \quad i \in I_k,\ b_{i,k}<b^{1:n_k},
	\end{equation}
	whenever the numerator is positive.  
	
	The empirical content depends on how $C_k$ is specified.  Here we outline our baseline specification, $C_k=\phi(Z_{k+1})$, which yields a valid bound under the following monotonicity condition on common values:
	\begin{condition}[Descending common values]\label{cond:descending}
		$\phi(Z_k) > \phi(Z_{k+1})$ for $k < K$.
	\end{condition}
	\noindent This condition is a transparent sufficient condition for deriving a recursive bound on $W_k(\theta_i)$. It should not be interpreted as a restriction on bidder preferences, market structure, or the auctioneer's ordering rule.\footnote{The earlier work of \cite{elmaghraby2003importance}, \cite{kittsteiner2004declining}, and \cite{shi2022implementing} relates item ordering to revenue or efficiency under specific equilibrium models; our use of Condition~\ref{cond:descending} is different in spirit, requiring no equilibrium model and no sorting rule.}  When bidders view the full remaining sequence as their relevant opportunity set, Condition~\ref{cond:descending} requires common values to be decreasing along that remaining sequence so that the value of waiting after period $k$ is bounded by the next opportunity. In practice, the condition can be imposed locally on any tail of the sequence: if $\phi(Z_k)>\phi(Z_{k+1})>\cdots>\phi(Z_K)$ for a given market, then the bounds below can be applied to the auctions in that decreasing tail even when earlier auctions in the same market do not satisfy the condition.
	
	\begin{lemma}[Bound on expected future payoff]\label{lem:expected_payoff_bound}
		Under Condition~\ref{cond:descending}, the expected future payoff in period $k$, $W_k(\theta_i)$, satisfies:
		\begin{enumerate}
			\item[(a)] $W_k(\theta_i) \leq \theta_i \phi(Z_{k+1})$ for all $i \in I_k$ and $k < K$.
			\item[(b)] Let $V_{i,t}\equiv \theta_i\phi(Z_t)-(p_t+\Delta)$ denote the maximum net payoff from winning period~$t$ at the minimum feasible payment, and set $V_{i,K+1}=0$.  If $V_{i,t}\geq \max\{0,V_{i,t+1}\}$ for all $t=k+1,\ldots,K$, then
			\begin{align}\label{eq:W_bound}
				W_k(\theta_i) \leq V_{i,k+1}
				= \theta_i \phi(Z_{k+1}) - (p_{k+1} + \Delta), \quad \forall\, i \in I_k,\ k < K.
			\end{align}
		\end{enumerate}
	\end{lemma}
	Part~(a), which always holds under Condition~\ref{cond:descending}, provides an upper bound on $W_k(\theta_i)$ with $C_k=\phi(Z_{k+1})$ and $m_k=0$. Part~(b) refines Part~(a) by incorporating opening prices along the relevant future tail with $m_k=p_{k+1}+\Delta$.\footnote{Since $W_k(\theta_i) \geq 0$ (bidders can always abstain from future auctions), Part~(b) is informative only for bidders whose relevant future net payoff is nonnegative. The condition $V_{i,t}\geq \max\{0,V_{i,t+1}\}$ is a sufficient recursive dominance condition: each future auction's maximum net payoff weakly exceeds both zero and the bound on later continuation payoffs. For adjacent periods with positive later net payoff, it requires $\theta_i\{\phi(Z_t)-\phi(Z_{t+1})\}\geq p_t-p_{t+1}$, so opening prices in both periods enter the comparison.} When the condition in Part~(b) holds, combining Lemma~\ref{lem:expected_payoff_bound} with \eqref{eq:opportunity_cost_ineq} yields
	\begin{align}\label{eq:pre_ub}
		\theta_i (\phi(Z_k) - \phi(Z_{k+1})) \leq b^{1:n_k} + \Delta - (p_{k+1} + \Delta) = b^{1:n_k} - p_{k+1}.
	\end{align}
	Since $\phi(Z_k) > \phi(Z_{k+1})$ by Condition~\ref{cond:descending}, dividing by the positive quantity $\phi(Z_k) - \phi(Z_{k+1})$ yields an upper bound on $\theta_i$:
	
	\begin{theorem}[Upper bound]\label{thm:ub_general}
		Under Assumption~\ref{ass:opportunity_cost} and Condition~\ref{cond:descending}, for all $i \in I_k$ with $b_{i,k} < b^{1:n_k}$ and $k < K$,
		\begin{align}
			\theta_i &\leq \frac{b^{1:n_k}+\Delta}{\phi(Z_k) - \phi(Z_{k+1})}, \label{eq:ub_general_gross}\\
			\theta_i &\leq \frac{b^{1:n_k} - p_{k+1}}{\phi(Z_k) - \phi(Z_{k+1})}
			\quad\text{if Lemma~\ref{lem:expected_payoff_bound}(b) applies.}\label{eq:ub_general}
		\end{align}
		For a winner with $b_{i,k}=b^{1:n_k}$, no tighter bound than $\bar{\theta}$ is available.
	\end{theorem}
	\noindent The first finite bound is called the baseline (gross) bound implied by Part~(a) of Lemma~\ref{lem:expected_payoff_bound}.  The second is referred to as the sharper net bound implied by Part~(b). When $k+2\leq K$, the first adjacent-period dominance condition is
	\[
	\theta_i\phi(Z_{k+1})-(p_{k+1}+\Delta)
	\geq
	\max\{0,\theta_i\phi(Z_{k+2})-(p_{k+2}+\Delta)\}.
	\]
	When the later net payoff is positive, this is equivalent to $\theta_i\{\phi(Z_{k+1})-\phi(Z_{k+2})\}\geq p_{k+1}-p_{k+2}$. Thus, if opening prices are zero and the relevant future net payoffs are nonnegative, Condition~\ref{cond:descending} makes the adjacent comparisons automatic; the condition is also automatic in the two-period case once the next auction is profitably enterable.
	
	The upper bounds in \eqref{eq:ub_last}, \eqref{eq:ub_general_gross}, and \eqref{eq:ub_general} are most informative for the second-highest order statistic $\theta^{2:n_k}$. The bound on $\theta^{2:n_k}$ implies the same inequality on all lower-order statistics since $\theta^{j:n_k} \leq \theta^{2:n_k}$ for $j > 2$. The bound for the highest-order statistic is uninformative (only $\bar{\theta}$). When Lemma~\ref{lem:expected_payoff_bound}(b) applies, the relevant upper bounds are:
		\begin{equation}\label{eq:ub_summary}
			\theta^{2:n_k} \leq \begin{cases}
				\displaystyle\frac{b^{1:n_k} - p_{k+1}}{\phi(Z_k) - \phi(Z_{k+1})}, & \text{if } k < K, \\[6pt]
				\displaystyle\frac{b^{1:n_K} + \Delta}{\phi(Z_K)}, & \text{if } k = K.
			\end{cases}
		\end{equation}

    \begin{remark}
    Without imposing Condition~\ref{cond:descending}, one can derive a more conservative upper bound on $W_k(\theta_i)$ by specifying $C_k=\max_{t\geq k+1}\phi(Z_t)$.\footnote{Other alternative specifications can be considered with limited attention or discounting, where $C_k=\max_{k<t\leq k+H}\phi(Z_t)$ or $C_k=\max_{t>k}\delta^{t-k}\phi(Z_t)$ for $\delta\in(0,1]$, respectively.} The identifying restriction is no longer global monotonicity of the sequence, but the observation-level filter $\phi(Z_k)>C_k$. Below, the distribution bounds are written for the baseline gross statistic; when Lemma~\ref{lem:expected_payoff_bound}(b) applies, the sharper net statistic is obtained by setting $m_k=p_{k+1}+\Delta$. When the monotonicity condition is not met, the valuation bounds can therefore be implemented with an alternative specification of $C_k$ and $m_k$.   
    \end{remark}

	\subsection{Bounds on the Valuation Distribution}\label{subsec:dist_bounds}
    
    
	We now translate the valuation bounds into bounds on the distribution $F^0$. Lemma~\ref{lem:lb} implies the first-order stochastic dominance:
	\begin{align}\label{eq:fosd_ub}
		F_{j:n_k}(x) \leq G_{j:n_k}(x), \quad \forall\, x \in [\underline{\theta}, \bar{\theta}],\ n_k \leq N,\ j \in \{1, \ldots, n_k\}.
	\end{align}
	The distribution of the $j$th order statistic from $n_k$ i.i.d.\ draws from $F^0$ can be expressed via the incomplete beta function:
	\begin{align}\label{eq:beta}
		F_{j:n_k}(x) = \Gamma(F^0(x);\, n_k - (j-1),\, j),
	\end{align}
	where $\Gamma(z; a, b) \equiv I_z(a,b) = \int_0^z \frac{(a+b-1)!}{(a-1)!(b-1)!} t^{a-1}(1-t)^{b-1}\, dt$ is the regularized incomplete beta function. Since $\Gamma^{-1}(\cdot)$ is monotone, inverting \eqref{eq:beta} and combining with \eqref{eq:fosd_ub} gives
	\begin{align}\label{eq:F_upper}
		F^0(x) \leq F_U(x) \equiv \min_{\substack{k \in \{1,\ldots,K\} \\ n_k \leq N,\ j \in \{1,\ldots,n_k\}}} \Gamma^{-1}\!\big(G_{j:n_k}(x);\, n_k - (j-1),\, j\big), \quad \forall\, x \in [\underline{\theta}, \bar{\theta}].
	\end{align}
	The minimum across all periods and order statistics selects the tightest upper bound for each $x$.
	
	For the lower bound on $F^0$, consider the last period. Define $\zeta_{1:n_K}^\Delta \equiv (b^{1:n_K} + \Delta)/\phi(Z_K)$ with distribution $G_{1:n_K}^\Delta(\cdot)$.\footnote{We use $\zeta$ rather than $\eta$ for this statistic to distinguish it from the reduced bid $\eta_{i,k} = b_{i,k}/\phi(Z_k)$; note that $\zeta_{1:n_K}^\Delta$ includes the bid increment~$\Delta$ and uses the winning bid rather than a single bidder's bid.} The upper bound on $\theta^{2:n_K}$ from \eqref{eq:ub_summary} yields:
	\begin{align}\label{eq:fosd_lb_last}
		F_{2:n_K}(x) \geq G_{1:n_K}^\Delta(x), \quad \forall\, x \in [\underline{\theta}, \bar{\theta}],\ n_K \leq N.
	\end{align}
	For periods $k < K$, define the baseline sequential statistic $h^g_{1:n_k} \equiv (b^{1:n_k}+\Delta)/(\phi(Z_k) - \phi(Z_{k+1}))$ with distribution $H^g_{1:n_k}(\cdot)$. Then:
	\begin{align}\label{eq:fosd_lb_k}
		F_{2:n_k}(x) \geq H^g_{1:n_k}(x), \quad \forall\, x \in [\underline{\theta}, \bar{\theta}],\ n_k \leq N.
	\end{align}
	Applying the beta inverse transformation and taking the maximum across all bounds:
	\begin{align}\label{eq:F_lower}
		F^0(x) \geq F_L(x) \equiv \max\!\left\{\max_{n_K \leq N} \Gamma^{-1}\!\big(G_{1:n_K}^\Delta(x);\, n_K\!-\!1,\, 2\big),\; \max_{\substack{k < K \\ n_k \leq N}} \Gamma^{-1}\!\big(H^g_{1:n_k}(x);\, n_k\!-\!1,\, 2\big)\right\},
	\end{align}
	for all $x \in [\underline{\theta}, \bar{\theta}]$. When Lemma~\ref{lem:expected_payoff_bound}(b) applies, the sharper net statistic $h^n_{1:n_k}\equiv (b^{1:n_k}-p_{k+1})/(\phi(Z_k)-\phi(Z_{k+1}))$ can replace $h^g_{1:n_k}$ in \eqref{eq:fosd_lb_k}--\eqref{eq:F_lower}. When opening prices vary across auctions, the bounds $F_U$ and $F_L$ apply to the unconditional truncated distribution $F(\cdot)$ defined in \eqref{eq:truncated_unconditional}. Conditional on $\tilde{p}$, order statistics and valuations are i.i.d.\ from the truncated distribution; unconditional moments average the corresponding conditional order-statistic probabilities over the distribution of $\tilde{p}$.
	
	\label{subsec:panel_bounds}
	When individual bidders can be tracked across periods in data, the panel structure provides an additional source of identifying information. For each bidder $i$ observed in periods $k \in \mathcal{K}_i$, Assumption~\ref{ass:no_overbid} implies $\eta_{i,k} = b_{i,k}/\phi(Z_k) \leq \theta_i$ for every period $k \in \mathcal{K}_i$. Taking the maximum across all observed periods:
	\begin{align}\label{eq:panel_lb}
		\eta_i^{\max} \equiv \max_{k \in \mathcal{K}_i} \frac{b_{i,k}}{\phi(Z_k)} \leq \theta_i.
	\end{align}
	Since $\eta_i^{\max} \leq \theta_i$ for every bidder $i$, the event $\{\theta_i \leq x\}$ implies $\{\eta_i^{\max} \leq x\}$. Therefore, with $G^{\max}(\cdot)$ denoting the distribution of $\eta_i^{\max}$:
	\begin{align}\label{eq:F_upper_panel}
		F^0(x) \leq G^{\max}(x), \quad \forall\, x \in [\underline{\theta}, \bar{\theta}].
	\end{align}
	This bound can be tighter than single-period beta-inversion bounds because it uses the maximum of a bidder's reduced bids across all auctions, not just one. Importantly, \eqref{eq:F_upper_panel} does not rely on the order statistics/beta function inversion, and does not require Condition~\ref{cond:descending}; it uses only Assumption~\ref{ass:no_overbid} applied across periods.
	
	\begin{remark}[Selection in later periods]\label{rem:selection}
		The distribution bounds in \eqref{eq:F_upper} and \eqref{eq:F_lower} are written for auctions whose participating bidders' private components are i.i.d.\ draws from $F^0$.  This is immediate in the first period and is also appropriate when each auction is interpreted as drawing a renewed active bidder set from a stable population.  With a fixed initial cohort, unit demand, and no new entry, later-period participation is selected because previous winners no longer bid.  In period $k$, at most $k-1$ bidders from the original cohort have exited.  Therefore the second-highest remaining value is weakly above the $(k+1)$st highest order statistic from the original $N$ draws: $\theta^{2:(N-k+1)} \geq \theta^{k+1:N}$. Since Theorem~\ref{thm:ub_general} gives $h^g_{1:n_k} \geq \theta^{2:(N-k+1)}$ for the remaining bidder set, we have $h^g_{1:n_k} \geq \theta^{k+1:N}$, and hence
		\[
		\Pr(h^g_{1:n_k} \leq x) \leq \Pr(\theta^{k+1:N}\leq x)
		= \Gamma(F^0(x);\, N-k,\, k+1).
		\]
		Thus later-period lower-bound moments can be mapped back to $F^0$ by using the beta shape $(N-k,k+1)$ rather than $(n_k-1,2)$.  The correction is exact for ranks when the highest remaining private value exits after each sale and conservative otherwise.  The no-overbidding upper bound remains conservative under the same fixed-cohort selection: the $j$th highest remaining value is weakly below the $j$th highest value in the original cohort, so $\eta^{j:\mathrm{rem},k}\leq \theta^{j:\mathrm{rem},k}\leq \theta^{j:N}$ and the original-cohort beta inversion remains valid.
	\end{remark}
		
	\subsection{Robustness of the Bounds}\label{subsec:robustness}
	
		The bounds can also accommodate several features commonly encountered in real auction markets as discussed below. Appendix~\ref{app:special_case} further discusses no-bid periods and shows how they can be removed when the remaining sequence satisfies the relevant descending filter.
	
	\begin{itemize}
		\item \textbf{Infinite horizon:}  Our baseline model assumes a finite number of periods~$K$.  In many auction markets, including online used car auctions where listings recur continuously, the horizon is effectively infinite ($K = \infty$) and bidders enter and exit stochastically. The expected future payoff for bidder $i$ with deadline $d_i$ becomes
		\begin{align}\label{eq:W_infinite}
			W_k(\theta_i, d_i) = \int_{n_k}^{\bar{N}} \bigg\{&P(b_{i,k+1} > \max_{j \neq i} b_{j,k+1} \mid n_{k+1})\, \mathbf{E}[\pi_{i,k+1} \mid \text{win}, n_{k+1}] \nonumber \\
			&+ \big[1 - P(b_{i,k+1} > \max_{j \neq i} b_{j,k+1} \mid n_{k+1})\big]\, W_{k+1}(\theta_i, d_i)\bigg\}\, dF_N(n_{k+1}),
		\end{align}
			for $k \leq d_i$, where $F_N$ is the law of future bidder counts and $W_k(\theta_i, d_i) = 0$ for $k > d_i$. The following lemma ensures that our pairwise bounds remain valid in long-horizon settings without requiring a terminal period.
		
		\begin{lemma}\label{lem:ub_infinite}
			Under Assumption~\ref{ass:opportunity_cost} and Condition~\ref{cond:descending}, Part~(a) of Lemma~\ref{lem:expected_payoff_bound} continues to hold under the payoff specification \eqref{eq:W_infinite}: $W_k(\theta_i, d_i) \leq \theta_i\phi(Z_{k+1})$ for all $k \leq d_i$.  Hence the baseline bound in Theorem~\ref{thm:ub_general} remains valid.  The sharper net refinement in Part~(b) applies under the same recursive dominance condition as in the finite-horizon case.
		\end{lemma}
		
		\item \textbf{Impatient bidders}:  When bidders discount future payoffs by a factor $\delta \in (0,1)$, the expected future payoff decreases.  Since our upper bound on $W_k(\theta_i)$ uses $W_k(\theta_i) \leq \theta_i\phi(Z_{k+1})$, discounting only tightens the inequality (the actual future payoff is smaller), so the bounds in Theorem~\ref{thm:ub_general} remain valid.
		
		\item \textbf{Noisy signals}:  If bidders observe only noisy signals about future items' characteristics, they know only the expected value of future objects at the time of bidding.  Since the inequalities bound the \emph{maximum possible} (most optimistic) expected future payoff, they remain valid when bidders have noisy information.	
	\end{itemize}

	\subsection{Sharp Bounds}\label{subsec:sharp}
	
	The bounds we derived so far are not necessarily sharp. We now characterize the sharp bounds using the generalized instrumental variable (GIV) framework of \cite{chesher2017generalized}. The value of this framework is that sharpness can be stated through random-set theory rather than via a constructive proof, which is difficult even in the static HT model.  To simplify exposition, consider the fixed-$N$ case with the baseline specification $C_k=\phi(Z_{k+1})$ and $m_k=0$ (zero opening prices and $\Delta = 0$). Let $\mathcal{K}$ denote the terminal period together with the retained non-terminal periods satisfying the adjacent sign condition $\phi(z_k)>\phi(z_{k+1})$; periods that fail this filter are omitted from the sharp-set construction.\footnote{Other continuation-index specifications are handled by replacing $b_k^{1:N}/\{\phi(Z_k)-\phi(Z_{k+1})\}$ below with the statistic in \eqref{eq:Ck_upper}.} For each $k\in\mathcal{K}$, define the vector of ordered bids $B_k \equiv (b_k^{1:N}, \ldots, b_k^{N:N})$ and item characteristics $Z_k$. Let $U \equiv (u_1, \ldots, u_N)$ denote the vector of transformed value order statistics, $u_j \equiv F^0(\theta^{j:N})$, so that $U$ is supported on $\mathcal{R}_U=\{1\geq u_1\geq\cdots\geq u_N\geq0\}$ with density $N!$.
	
	Define $B \equiv (B_k)_{k\in\mathcal{K}}$ and $Z \equiv (Z_k)_{k\in\mathcal{K}}$.  The inequalities \eqref{eq:value_lb} and \eqref{eq:ub_summary} imply the structural function $h(B,Z,U)=\sum_{k\in\mathcal{K}} h_k(B_k,Z_k,U)=0$ a.s., where
	\begin{equation}\label{eq:structural}
		h_k(B_k, Z_k, U) = \max\!\left(u_2 - F^0\!\left(\frac{b_k^{1:N}}{\phi(Z_k) - \phi(Z_{k+1})}\right), 0\right) + \sum_{j=1}^N \max\!\left(F^0\!\left(\frac{b_k^{j:N}}{\phi(Z_k)}\right) - u_j,\, 0\right),
	\end{equation}
	with $\phi(Z_{K+1})\equiv0$. For a closed subset $S\subseteq\mathcal{R}_U$, let $G_U(S)\equiv P[U\in S]$. Given observables $B=b$ and $Z=z$, the set of latent variables consistent with the model is
	\begin{equation}\label{eq:U_level_set}
		\mathcal{U}(b, z; h) \equiv \left\{u \in \mathcal{R}_U : \forall\, k \in \mathcal{K},\ F^0\!\left(\frac{b_k^{1:N}}{\phi(z_k) - \phi(z_{k+1})}\right) \geq u_2 \text{ and } F^0\!\left(\frac{b_k^{j:N}}{\phi(z_k)}\right) \leq u_j,\ \forall\, j\right\}.
	\end{equation}
	When $B$ is random, $\mathcal{U}(B,z;h)$ is a random set.  Let $\mathcal{Q}(z;h)\equiv\{\mathcal{U}(b,z;h): b\in\mathcal{R}_B\}$ be its conditional support and $\mathcal{Q}^*(z;h)$ denote the collection of unions of elements of $\mathcal{Q}(z;h)$.
	
	\begin{lemma}[Sharp bounds]\label{lem:sharp}
		Under Assumptions~\ref{ass:no_overbid}--\ref{ass:opportunity_cost}, the identified set $\mathcal{H}^*$ of the structural function $h$ is
		\begin{equation}\label{eq:sharp_bounds}
			\mathcal{H}^* \equiv \left\{h : \forall\, S \in \mathcal{Q}^*(z; h),\ P[\mathcal{U}(B, z; h) \subseteq S \mid Z = z] \leq G_U(S)\ \text{for a.e. } z\right\}.
		\end{equation}
	\end{lemma}
	
	The sharp set therefore consists of all valuation distributions whose induced structural function satisfies an uncountable collection of Artstein inequalities. The bounds in \eqref{eq:F_upper} and \eqref{eq:F_lower} correspond to particular contiguous unions of $U$-level sets. To see this, write reduced bids as $\eta_k=(\eta_{k,1},\ldots,\eta_{k,N})$ with $b_k(\eta_k)=(\eta_{k,1}\phi(z_k),\ldots,\eta_{k,N}\phi(z_k))$ and $\eta_{k,1}\geq\cdots\geq\eta_{k,N}$, and define $b(\eta)\equiv(b_k(\eta_k))_{k\in\mathcal{K}}$. For two arrays $\eta'$ and $\eta''$, define
	\begin{equation}\label{eq:contiguous_union}
		S((\eta',\eta''),z;h)\equiv \bigcup_{\eta\in[\eta',\eta'']}\mathcal{U}(b(\eta),z;h),
	\end{equation}
	where $[\eta',\eta'']$ is the Cartesian product of componentwise intervals. For the class of intervals used below, this union occupies
	\begin{multline}\label{eq:contiguous_region}
		S((\eta',\eta''),z;h)=
		\Bigg\{u\in\mathcal{R}_U:\ \forall k\in\mathcal{K},\ 
		F^0\!\left(\frac{\eta_{k,1}''\phi(z_k)}{\phi(z_k)-\phi(z_{k+1})}\right)\geq u_2,\\
		\text{and}\quad F^0(\eta_{k,j}')\leq u_j,\ \forall j=1,\ldots,N
		\Bigg\}.
	\end{multline}
	Applying Lemma~\ref{lem:sharp} to these contiguous unions yields moment inequalities that every admissible $F^0$ must satisfy.
	
	The no-overbidding upper bound is recovered by fixing a period $k$ and rank $m$, setting $\eta_{k,1}''=\bar{\theta}$, $\eta_{k,j}'=\underline{\theta}$ for $j>m$, and $\eta_{k,j}'=\theta$ for $j\leq m$, while leaving the restrictions for other periods uninformative. The resulting inequality is
	\begin{equation}\label{eq:giv_upper}
		\Gamma(F^0(\theta);\, N-m+1,\, m) \leq G_{k,m:N}(\theta),
	\end{equation}
	where $G_{k,m:N}$ is the distribution of the reduced-bid order statistic $\eta_k^{m:N}$. Since $\Gamma(\cdot;N-m+1,m)$ is increasing,
	\[
		F^0(\theta) \leq \Gamma^{-1}\!\big(G_{k,m:N}(\theta);\, N-m+1,\, m\big), \quad \forall k\leq K,
	\]
	which gives the upper bound in \eqref{eq:F_upper} after taking the minimum over periods and order statistics.
	
	The lower bounds are recovered analogously. For the terminal period, set $\eta_{K,1}''=\theta$ and leave all other restrictions uninformative. Then
	\begin{equation}\label{eq:giv_lower_last}
		G_{K,1:N}(\theta)\leq \Gamma(F^0(\theta);\, N-1,\,2).
	\end{equation}
	For a non-terminal period $k<K$, set $\eta_{k,1}''=\theta\{\phi(z_k)-\phi(z_{k+1})\}/\phi(z_k)$ and again leave other restrictions uninformative. This gives
	\begin{equation}\label{eq:giv_lower_seq}
		H_{k,1:N}(\theta)\leq \Gamma(F^0(\theta);\, N-1,\,2),
	\end{equation}
	where $H_{k,1:N}$ is the distribution of $\eta_k^{1:N}\phi(z_k)/\{\phi(z_k)-\phi(z_{k+1})\}$. Inverting \eqref{eq:giv_lower_last} and \eqref{eq:giv_lower_seq} gives the lower bound in \eqref{eq:F_lower}.
	
	These examples show that \eqref{eq:F_upper} and \eqref{eq:F_lower} are necessary implications of the sharp GIV characterization, but they do not exhaust it. Additional inequalities combine restrictions across order statistics and periods. For instance, for $a,b\in[\underline{\theta},\bar{\theta}]$ with $b\leq a$ and $a\phi(z_k)/\{\phi(z_k)-\phi(z_{k+1})\}\in[\underline{\theta},\bar{\theta}]$,
	\begin{equation}\label{eq:additional_ineq}
		P\!\left(\eta_k^{1:N} \leq a,\ \eta_k^{2:N} \geq b\right) \leq \Gamma\!\left(F^0\!\left(\frac{a\phi(z_k)}{\phi(z_k) - \phi(z_{k+1})}\right);\, N\!-\!1, 2\right) - \Gamma\!\left(F^0(b);\, N\!-\!1, 2\right).
	\end{equation}
	This inequality restricts the increase in $F^0$ between $b$ and $a\phi(z_k)/\{\phi(z_k)-\phi(z_{k+1})\}$. Since the sharp characterization involves an uncountable collection of such inequalities, implementation is computationally demanding. The empirical analysis therefore uses the tractable non-sharp bounds, while the GIV characterization clarifies exactly where additional identifying content resides.

	\section{Estimation and Inference}\label{sec:estimation_inference}
	
	This section addresses the statistical challenges of estimating the bounds from finite samples and constructing confidence intervals for the partially identified distribution $F^0$.
	
	\subsection{Estimation when $N$ is fixed}\label{subsec:finite_sample}
	
	When each auction has the same number of bidders $N$, the distribution bounds in \eqref{eq:F_upper} and \eqref{eq:F_lower} are computed by inverting the incomplete beta function $\Gamma^{-1}(\hat{G}_{j:N}(x);\, a,\, b)$, where $\hat{G}_{j:N}$ is the empirical CDF of the observed $j$th-order statistic.  This inversion is well-behaved in the population but exhibits two fundamental difficulties in finite samples. First, the mapping from order-statistic distributions to the parent distribution is not Lipschitz continuous at the support boundaries.  \citet{menzel2013large} show that the derivative of $\Gamma^{-1}$ is unbounded at the endpoints, so that small estimation errors in $\hat{G}$ are amplified without bound near the tails.  The resulting optimal convergence rate for nonparametric estimation of $F^0$ from the $j$th-order statistic is as slow as $T^{-1/N}$, where $T$ is the number of auctions, dramatically slower than the usual $T^{-2/5}$ rate for density estimation.  \citet{cherapanamjeri2022estimation} further establish that full-support recovery of $F^0$ in Kolmogorov distance is impossible from order statistics, with the exponential dependence on $N$ constituting a fundamental lower bound that no estimator can improve upon. Second, the empirical CDF $\hat{G}_{j:N}$ is degenerate in the left tail for high order statistics.  Since $\Pr(\theta_{j:N} \leq v) = \Gamma(F^0(v);\, N\!-\!j\!+\!1,\, j) \approx [F^0(v)]^{N-j+1}$, which vanishes rapidly for small $F^0(v)$, the ECDF of the $j$th-order statistic has zero mass below a cutoff that grows with $N-j+1$.  Inverting $\hat{G}_{j:N}(v) = 0$ gives $\Gamma^{-1}(0) = 0$, collapsing the bound to zero.  This creates a ``dead zone'' in the left tail of $F_U$ that widens with the number of bidders, precisely the region where the bounds should be widest. These limitations are shared by all nonparametric bounds approaches based on order statistic inversion, including the HT framework.
	
	Equation \eqref{eq:F_upper} defines the population upper bound as the minimum of $\Gamma^{-1}(G_{j:N}(x);\, N\!-\!j\!+\!1,\, j)$ across all order statistics $j = 1, \ldots, N$.  Each order statistic provides the strongest identification power at different quantiles: the second-order statistic ($j = 2$) is most informative in the upper tail but degenerate in the lower tail, while the minimum ($j = N$) is informative in the lower tail but quickly becomes degenerate as one approaches the right tail. In the population, the intersection (minimum) across all $j$ yields the sharpest bound.  In finite samples, however, the na\"ive sample analogue (taking the pointwise minimum of $\Gamma^{-1}(\hat{G}_{j:N})$ across $j$) inherits the worst finite-sample pathology of each $j$, since the minimum of noisy estimators is biased downward \citep{chernozhukov2013intersection}.  The original HT estimator addresses finite-sample extremum bias by replacing the pointwise min/max across order-statistic inequalities with a smooth exponential-weighted approximation. Rather than implement this ad hoc smoothing device, we include the bias-corrected intersection estimator of \citet{chernozhukov2013intersection}, which provides a systematic modern alternative for the same intersection-bounds problem. We consider the following six alternative approaches to address this problem:
	
	\begin{enumerate}
		\item \textbf{Baseline ($j=2$ only):} Order-statistic inversion to the transaction-price rank only. This is a conservative finite-sample benchmark that avoids the extremum bias from combining multiple noisy order-statistic bounds, while sacrificing left-tail information that lower bid ranks can provide.
		
		\item \textbf{Na\"ive intersection:}  Pointwise minimum across all $j=1,\ldots,N$.  This is the direct sample analogue of \eqref{eq:F_upper} and serves as a worst-case reference for the finite-sample pathologies described above.
		
		\item \textbf{CLR bias-corrected intersection:}  The half-median bias correction of \citet{chernozhukov2013intersection}.  For each grid point $v$, the na\"ive minimum is shifted upward by the median bootstrap bias: $\hat{F}_U^{\text{CLR}}(v) = 2\hat{F}_U^{\min}(v) - \text{median}(\hat{F}_U^{*\min}(v))$, where the median is taken over bootstrap replications.
		
		\item \textbf{Bernstein polynomial smoothing:}  Following \citet{leblanc2012estimating}, the step-function ECDF $\hat{G}_{j:N}$ is replaced with a Bernstein polynomial estimator of degree $m = \lceil T^{2/5} \rceil$.  The smoothed CDF is strictly positive near the support boundary, eliminating the $\Gamma^{-1}(0) = 0$ problem.
		
		\item \textbf{Adaptive $j$-selection:}  At each grid point $v$, order statistics $j$ are included in the intersection only if they have sufficient local support: (i) $\hat{G}_{j:N}(v) \geq 1/T$ and (ii) at least five observations fall within a neighborhood of $v$.  This excludes degenerate bounds in the tails while preserving validity wherever the ECDF is well-estimated.
		
		\item \textbf{Inverse-variance weighted combination:}  Following \citet{armstrong2014weighted}, the bound at each $v$ is a weighted average $\hat{F}_U(v) = \sum_j w_j(v) \hat{F}_U^j(v) / \sum_j w_j(v)$, with weights $w_j(v) = 1/\widehat{\text{Var}}(\hat{F}_U^j(v))$ derived from the delta method.  This automatically downweights uninformative order statistics.
	\end{enumerate}
	
	\noindent We verify these approaches' finite-sample performance in Monte Carlo simulations in Section \ref{sec:sim}.
	
	\subsection{Estimation when $N$ varies across auctions}
	
		When the number of bidders $N_a$ varies across auctions, as is typical in practice, the beta inversion \eqref{eq:F_upper} cannot be applied directly with a common~$N$.  A natural estimator pools all auctions while respecting each auction's individual~$N_a$. For the $j$th order statistic, the no-overbidding inequality implies the population moment
		\[
		\mathbf{E}\!\left[\mathbf{1}\{\eta^{j:N_a}\leq v\}\right]
		\geq
		\mathbf{E}\!\left[I_{F^0(v)}(N_a-j+1,j)\right],
		\]
		where $I_x(a,b) \equiv \Gamma(x; a, b)$ is the regularized incomplete beta function defined in \eqref{eq:beta}. The estimator replaces expectations by sample analogs and solves
		\begin{equation}\label{eq:moment_condition}
			\frac{1}{T}\sum_{a=1}^T \mathbf{1}\{\eta_a^{j:N_a} \leq v\}
			=
			\frac{1}{T}\sum_{a=1}^T I_{F(v)}(N_a - j + 1,\, j).
		\end{equation}
		The upper bound $F_U(v)$ is the solution to the equation obtained by replacing the inequality with equality.  Since $I_F(a,b)$ is monotone increasing in~$F$ for each~$(a,b)$, the right-hand side is monotone in~$F$ and the equation can be solved by bisection at each grid point~$v$.
		
		The lower bound has an analogous moment condition.  By Theorem~\ref{thm:ub_general}, the baseline sequential statistic $h^{g,a}_{1:N_a} = (b^{1:N_a}_a+\Delta)/(\phi(Z_{k,a}) - \phi(Z_{k+1,a}))$ bounds every non-winner valuation from above; since only the winner is unrestricted, it follows that $h^{g,a}_{1:N_a}\geq\theta^{2:N_a}$. Hence
		\[
		\mathbf{E}\!\left[\mathbf{1}\{h^g_{1:N_a}\leq v\}\right]
		\leq
		\mathbf{E}\!\left[I_{F^0(v)}(N_a-1,2)\right].
		\]
		The sample analog solves
		\begin{equation}\label{eq:moment_condition_LB}
			\frac{1}{T}\sum_{a=1}^T \mathbf{1}\{h^{g,a}_{1:N_a} \leq v\}
			=
			\frac{1}{T}\sum_{a=1}^T I_{F(v)}(N_a - 1,\, 2).
		\end{equation}
		The lower bound $F_L(v)$ is the solution to the equation obtained by replacing the inequality with equality; since the right-hand side is strictly monotone increasing in~$F$, the inequality direction is preserved and $F_L(v) \leq F^0(v)$. If Lemma~\ref{lem:expected_payoff_bound}(b) applies, the same moment condition can be evaluated with the sharper net statistic from \eqref{eq:ub_general}.
		
		For the upper bound, the no-overbidding inequality $\eta^{j:N_a} \leq \theta^{j:N_a}$ holds for every $j \in \{1, \ldots, N_a\}$, so any rank could in principle be used in \eqref{eq:moment_condition} with the corresponding beta shape $(N_a\!-\!j\!+\!1, j)$.  We restrict to $j = 2$ in our implementation for three reasons.  First, $j = 1$ is essentially uninformative: the regularized incomplete beta $I_F(N_a, 1) = F^{N_a}$ is close to zero except for $F$ very near one, so inverting the winning-bid ECDF yields $F_U(v) \approx 1$ across the bulk of the support.  Second, ranks $j \geq 3$ require beta shapes $(N_a - j + 1, j)$ that differ from the shape pinned down by the lower bound, $(N_a\!-\!1, 2)$.  Combining bounds across mismatched shapes can produce $F_U < F_L$ in finite samples and complicates the construction of a coherent confidence interval; using $j = 2$ matches the lower bound's beta shape and helps preserve coherence. Finite-sample crossings can still occur when endpoints are computed from different samples or multiple moments are combined.  Third, with heterogeneous $N_a$, the variation in $N_a$ already plays the role that variation in $j$ would play under fixed $N$: low-$N_a$ auctions have rapidly-rising beta weights $I_F(N_a\!-\!1, 2)$ in the left tail, supplying the left-tail identification that the adaptive multi-$j$ method of Section~\ref{subsec:finite_sample} is designed to recover.  The adaptive method is therefore unnecessary in our empirical applications, and we use a single $j = 2$ throughout.
	
	This \emph{moment-condition inversion} follows the general approach to estimation under partial identification with moment inequalities \citep{chernozhukov2007estimation}, adapted to the specific structure of order statistic bounds with heterogeneous~$N_a$.  In the original application of \citet{haile2003inference}, bounds are estimated by conditioning on~$N$ via kernel smoothing, restricting the estimation window to auctions with similar numbers of bidders; as HT note, allowing too much heterogeneity in~$N$ within the bandwidth causes the bounds to cross.  Our approach takes a different strategy: rather than conditioning on~$N$, it pools all auctions and lets the beta function weights $I_F(N_a\!-\!1, 2)$ absorb the heterogeneity in~$N_a$ directly.  This avoids both the information loss from narrow-bandwidth conditioning and the crossing problem from wide-bandwidth smoothing.
	
	The approach has two advantages over the fixed-$N$ approach.  First, it uses all $T$ auctions simultaneously with their individual~$N_a$, avoiding the information loss from partitioning into per-$N$ subsamples.  Second, heterogeneity in~$N_a$ provides identification across the full support of~$F^0$ through a natural weighting mechanism.  Each auction~$a$ contributes a beta CDF $I_{F(v)}(N_a\!-\!1, 2)$ to the right-hand side of \eqref{eq:moment_condition}.  The shape of this beta CDF depends on~$N_a$: for a small-$N$ auction (e.g., $N_a = 3$), $I_F(2, 2) = 3F^2 - 2F^3$ rises at much smaller values of~$F$ than high-$N$ beta CDFs, contributing information in the left tail of the distribution.  For a large-$N$ auction (e.g., $N_a = 30$), $I_F(29, 2) = 30F^{29} - 29F^{30}$ is essentially zero until~$F$ is close to~1, contributing information only in the right tail.  The bisection algorithm finds the single $F(v)$ that makes the average of these heterogeneous beta CDFs equal the observed empirical CDF.  In effect, low-$N$ auctions receive higher ``weight'' in the left tail and high-$N$ auctions receive higher weight in the right tail, with the weighting determined endogenously by the beta function shapes.  The variation in~$N_a$ thus plays the same role as variation in the order statistic rank~$j$ in the adaptive method, providing a natural alternative to combining multiple order statistics.
	
	The following proposition establishes validity, consistency, and asymptotic normality of the moment-condition inversion estimator.
	
	\begin{prop}[Validity, consistency, and asymptotic normality]\label{prop:moment_cond}
		Suppose auctions are i.i.d., $F^0$ is continuous, and $N_a$ is bounded.  Let $\bar{I}(F) \equiv \mathbf{E}[I_F(N_a\!-\!j\!+\!1, j)]$.
		\begin{enumerate}
			\item[(i)] \emph{Upper bound.}  Let $\bar{G}_U(v) \equiv \mathbf{E}[\mathbf{1}\{\eta^{j:N_a} \leq v\}]$.  The no-overbidding inequality implies
			\begin{equation}\label{eq:moment_population}
				\bar{G}_U(v) \geq \bar{I}(F^0(v)),
			\end{equation}
			and since $\bar{I}(F)$ is strictly increasing in~$F$, the unique solution $F_U^*(v)$ to $\bar{I}(F) = \bar{G}_U(v)$ satisfies $F_U^*(v) \geq F^0(v)$.  The estimator is consistent, for each $v$, $|\hat{F}_U(v) - F_U^*(v)| = O_p(T^{-1/2})$, with influence-function representation
			\[
			\hat{F}_U(v) - F_U^*(v)
			= \frac{T^{-1}\sum_{a=1}^T \left[\mathbf{1}\{\eta_a^{j:N_a}\leq v\}-I_{F_U^*(v)}(N_a\!-\!j\!+\!1,j)\right]}{\bar{I}'(F_U^*(v))} + o_p(T^{-1/2}),
			\]
				where $\bar{I}'(F) = \mathbf{E}[f_{\mathrm{Beta}}(F;\, N_a\!-\!j\!+\!1, j)]$, provided $\bar{I}'(F_U^*(v)) > 0$.  Combining with the central limit theorem yields
			\begin{equation}\label{eq:asymptotic_dist}
					\sqrt{T}\!\left(\hat{F}_U(v) - F_U^*(v)\right) \xrightarrow{d} \mathcal{N}\!\left(0,\; \frac{\operatorname{Var}\!\left[\mathbf{1}\{\eta^{j:N_a} \leq v\} - I_{F_U^*(v)}(N_a\!-\!j\!+\!1, j)\right]}{\left(\mathbf{E}\!\left[f_{\mathrm{Beta}}(F_U^*(v);\, N_a\!-\!j\!+\!1, j)\right]\right)^2}\right).
			\end{equation}
			When $N_a$ is treated as fixed, the numerator simplifies to $\bar{G}_U(v)(1 - \bar{G}_U(v))$ and the asymptotic variance reduces to $\hat\sigma(v)^2$ in \eqref{eq:analytical_se}.
			
			\item[(ii)] \emph{Lower bound.}  Let $\bar{G}_L(v) \equiv \mathbf{E}[\mathbf{1}\{h^g_{1:N_a} \leq v\}]$, with $\bar I$ evaluated at the lower-bound shape $(N_a\!-\!1, 2)$.  The sequential bound (Theorem~\ref{thm:ub_general}) implies
			\begin{equation}\label{eq:moment_population_LB}
				\bar{G}_L(v) \leq \bar{I}(F^0(v)),
			\end{equation}
			and the unique solution $F_L^*(v)$ to $\bar{I}(F) = \bar{G}_L(v)$ satisfies $F_L^*(v) \leq F^0(v)$.  The estimator is consistent, for each $v$, $|\hat{F}_L(v) - F_L^*(v)| = O_p(T^{-1/2})$, with the analogous influence-function representation, and
			\[
				\sqrt{T}\!\left(\hat{F}_L(v) - F_L^*(v)\right) \xrightarrow{d} \mathcal{N}\!\left(0,\; \frac{\operatorname{Var}\!\left[\mathbf{1}\{h^g_{1:N_a} \leq v\} - I_{F_L^*(v)}(N_a\!-\!1, 2)\right]}{\left(\mathbf{E}\!\left[f_{\mathrm{Beta}}(F_L^*(v);\, N_a\!-\!1, 2)\right]\right)^2}\right),
			\]
				with the same fixed-$N_a$ simplification.  The joint limit of $(\hat F_U,\hat F_L)$ is asymptotically normal with covariance determined by the covariance of the corresponding influence functions.  The covariance is zero when the upper- and lower-bound moments are computed from independent samples; when they are computed from overlapping auctions or observations clustered within the same market, it can be estimated by a paired or clustered bootstrap.
		\end{enumerate}
	\end{prop}
	
	\subsection{Inference}\label{subsec:inference}
	
		The valuation distribution $F^0$ is partially identified, lying between the population bounds $[F_L(v),F_U(v)]$ for each $v$. Standard pointwise confidence intervals around either estimated bound are insufficient because they account only for sampling error in that endpoint, not for the identification uncertainty represented by the width of the bounds. We construct confidence intervals using two complementary approaches. The first approach is the nonparametric bootstrap.  We resample auctions with replacement $B$ times, recomputing the bounds for each bootstrap sample.\footnote{When upper- and lower-bound statistics are computed from distinct pooled samples, each sample can be resampled using its own observation unit.  When the same auction or market can contribute to both endpoints, a paired or clustered bootstrap at the auction/market level preserves the covariance between the two endpoints.}  The pointwise $(1-\alpha)$ bootstrap interval uses Bonferroni endpoint quantiles:
		\[
		\hat{F}_U^{\text{conf}}(v) = Q_{1-\alpha/2}\!\left(\hat{F}_U^{*b}(v);\, b = 1,\ldots,B\right),
		\]
		and $\hat{F}_L^{\text{conf}}(v) = Q_{\alpha/2}(\hat{F}_L^{*b}(v))$ for the lower bound.  The confidence bounds $[\hat{F}_L^{\text{conf}}, \hat{F}_U^{\text{conf}}]$ are wider than the estimated bounds, and $\Pr(\hat{F}_L^{\text{conf}}(v) \leq F^0(v) \leq \hat{F}_U^{\text{conf}}(v)) \geq 1 - \alpha$ at each~$v$.
	
	\citet{imbens2004confidence} propose confidence intervals for a parameter $\theta \in [\theta_L, \theta_U]$ that cover the true $\theta$ (not the identified set) with asymptotic probability at least $1-\alpha$.  The CI is
	\[
	\text{CI}_\alpha = \left[\hat{F}_L(v) - c_n \cdot \hat{\sigma}_L(v),\;\; \hat{F}_U(v) + c_n \cdot \hat{\sigma}_U(v)\right],
	\]
	where $\hat{\sigma}_L$ and $\hat{\sigma}_U$ are standard errors of the bound estimators, computed analytically via \eqref{eq:analytical_se}, and $c_n$ solves $\Phi(c_n + \hat{\Delta}_n / \max(\hat{\sigma}_L, \hat{\sigma}_U)) - \Phi(-c_n) = 1 - \alpha$, with $\hat{\Delta}_n = \max(\hat{F}_U - \hat{F}_L, 0)$.  When the bounds are wide, $c_n \to z_{1-\alpha}$ (one-sided on each end); when they collapse to a point, $c_n \to z_{1-\alpha/2}$ (two-sided).
	
		For the moment-condition estimator \eqref{eq:moment_condition}, the standard errors $\hat\sigma_U$ and $\hat\sigma_L$ can be computed analytically via the delta method applied to the implicit equation $m(\hat F, v) = 0$. Let $Y_a$ denote the statistic entering the relevant endpoint, and let $(s_{1a},s_2)$ denote the corresponding beta shape: $(N_a-j+1,j)$ for the upper bound and $(N_a-1,2)$ for the lower bound.  With $m(F,v)\equiv T^{-1}\sum_a[I_F(s_{1a},s_2)-\mathbf{1}\{Y_a\leq v\}]$, the implicit function theorem gives
		\begin{equation}\label{eq:analytical_se}
			\hat\sigma(v) =
			\frac{
				\sqrt{T^{-1}\widehat{\operatorname{Var}}_a\!\left(\mathbf{1}\{Y_a\leq v\}-I_{\hat F(v)}(s_{1a},s_2)\right)}
			}{
				T^{-1}\sum_{a=1}^T f_{\text{Beta}}(\hat F(v);\, s_{1a},\, s_2)
			},
		\end{equation}
		where $f_{\text{Beta}}(\cdot; a, b)$ is the beta density.  The numerator is the standard error of the moment evaluated at the estimated bound; the denominator converts moment uncertainty into $F$-space via the average beta density. When $N_a$ is fixed, \eqref{eq:analytical_se} reduces to the familiar expression $\sqrt{\hat G(v)(1-\hat G(v))/T}/f_{\text{Beta}}(\hat F(v);s_1,s_2)$. We use the analytical formula for the reported Imbens--Manski bands because it avoids bootstrap noise and produces smooth confidence bands, while treating the bootstrap as the more conservative finite-sample benchmark. Both approaches provide pointwise coverage at each $v$.  For uniform coverage across the entire support, one could calibrate via the sup-$t$ bootstrap \citep{chernozhukov2007estimation}, at the cost of wider bands.
	
	\section{Simulation}\label{sec:sim}
	
	We evaluate the finite-sample performance of the six estimation approaches with fixed $N$, the moment-condition inversion method with varying $N_a$, and the inference procedures from the previous section. We also compare the proposed lower bounds to three alternatives: (i) the static bounds of \citet{haile2003inference} applied naively to each period (HT); (ii) the last-period-only bounds, which discard all non-terminal auctions; and (iii) the bidder-level panel bounds from Section~\ref{subsec:panel_bounds}.
	
	\subsection{Design and Bound Comparison with fixed $N$}\label{subsec:lower_bound_comparison}
	
	We generate bid data from a structural two-period sequential English auction model. In each simulation round, private values $\theta_i$ for $N = 10$ bidders are drawn i.i.d.\ from $\text{LogNormal}(\mu = 1, \sigma = 0.5)$. Common values are drawn independently across auctions. Bids are generated from the symmetric Bayesian Nash equilibrium of the sequential second-price auction \citep{kittsteiner2004declining}. In the second (last) period, each remaining bidder's maximum willingness to pay is $\text{MWP}_{i,2} = \phi(Z_2)\theta_i$. In the first period, forward-looking bidders account for the option value of winning in the future:
	\begin{align*}
		\text{MWP}_{i,1} = \phi(Z_1)\theta_i - \tau\phi(Z_2)\int_{\max(0,\, p_2/\phi(Z_2))}^{\theta_i} F^0(x)^{N-2}\, dx,
	\end{align*}
	where the integral represents the expected continuation value, the expected gain from participating in the second auction conditional on being the strongest remaining bidder. The lower integration limit accounts for the reserve price: bidders with $\theta_i < p_2/\phi(Z_2)$ cannot profitably enter the second auction and have zero continuation value.
	
	We simulate auctions starting at an opening price $p_k$. Bidders drop out when the price reaches their MWP. To introduce departures from equilibrium behavior, bidders ranked third or lower (by MWP) drop out at a random fraction $\omega_i \sim \text{Uniform}[0.3, 1]$ of their MWP, creating substantial noise in observed lower-ranked bids. The second-highest bidder drops at their full MWP, and the winner bids one increment $\Delta = 0.05$ above the second-highest dropout price. We consider three configurations.  In all scenarios, the opening price in period~2 is stochastic: $p_2 = \phi(Z_2) \cdot U(0.5, 1)$, so that the reserve price scales with the item's common value.
	\begin{enumerate}[(a)]
		\item \emph{Large drop}: $\phi(Z_1) = 1$, $\phi(Z_2) \sim \text{Uniform}(0.01, 0.2)$, $\tau = 1$. The opportunity cost of winning period~1 is minimal.
		\item \emph{Small drop}: $\phi(Z_1) = 1$, $\phi(Z_2) \sim \text{Uniform}(0.7, 0.9)$, $\tau = 1$. The opportunity cost is substantial.
		\item \emph{Supply uncertainty}: $\phi(Z_1) = 1$, $\phi(Z_2) \sim \text{Uniform}(0.7, 0.9)$, $\tau = 0.5$. Uncertainty about whether a second auction occurs reduces the continuation value.
	\end{enumerate}
	
	For each scenario, we run 500 Monte Carlo replications; in each replication we generate 500 sequential auction pairs, compute all bounds using the full set of observed bids, and average the bounds and coverage statistics across replications. We first examine the performance of the lower bounds, where our approach differs most from the alternatives. Figure~\ref{fig:sim} displays the mean bounds for the three scenarios. In all three scenarios, the HT lower bound (in the middle row of Figure~\ref{fig:sim}) lies above the no-overbidding upper bound over a substantial portion of the support. This occurs because HT treats first-period bids as if they were stand-alone valuations, producing inflated lower bound statistics. By contrast, the proposed upper and lower bounds (in the top row of Figure~\ref{fig:sim}) bracket the true $F^0$ across the bulk of the support in the large-drop scenario. In the small-drop and supply-uncertainty scenarios, the bounds are extremely narrow, so that even minor finite-sample noise can cause the bounds to cross at isolated grid points. The stochastic opening price $p_2 = \phi(Z_2) \cdot U(0.5, 1)$ plays a key role in identification. In the sequential bound $\theta_i \leq (b^{1:n_1} - p_2)/(\phi_1 - \phi_2)$, a positive~$p_2$ reduces the numerator, tightening the upper bound on valuations and hence the lower bound on~$F$. In the simulations using this sharper net statistic, we retain observations satisfying $b^{1:n_1}/\phi_1 > p_2/\phi_2$. Simultaneously, the positive reserve screens out low-value bidders from period~2, weakening the last-period bound. The sequential structure therefore provides the greatest identification power when opening prices are non-trivial.
	
	\begin{figure}[p]
		\centering
		\includegraphics[width=\textwidth]{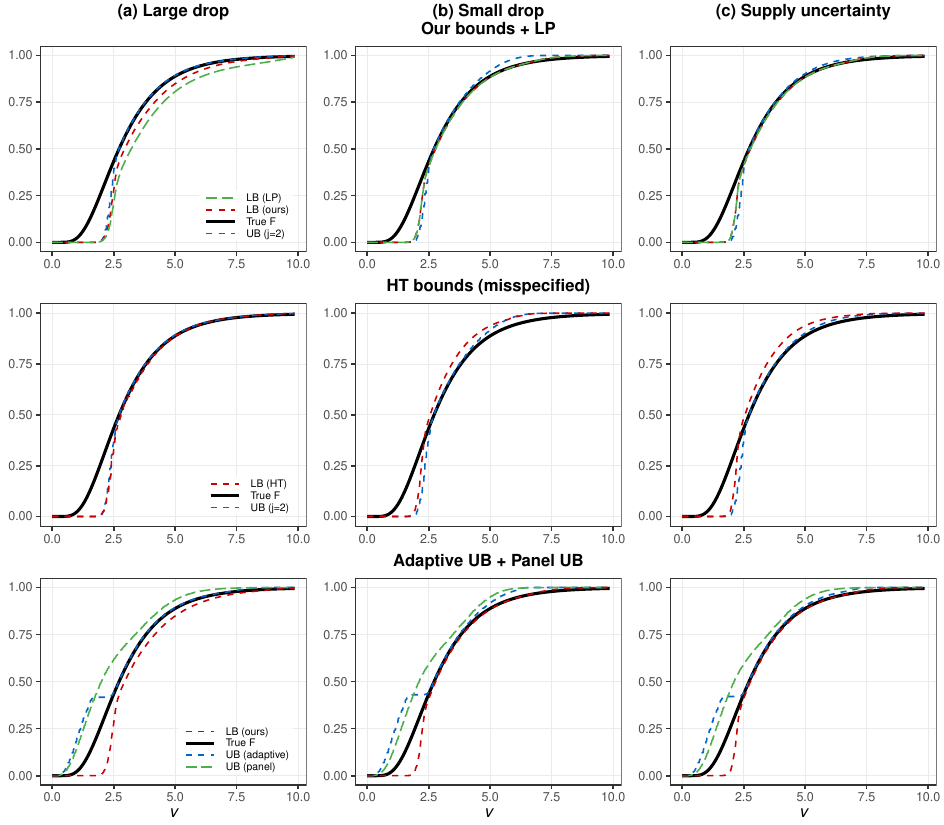}
		\caption{Mean bounds across 500 Monte Carlo iterations of 500 sequential auctions ($N = 10$, $\Delta = 0.05$, $p_2 = \phi(Z_2) \cdot U(0.5, 1)$).  Columns correspond to the three scenarios; rows show different bound combinations.  \emph{Top row}: Our sequential bounds, upper bound (UB $j\!=\!2$, blue dashed) and lower bound (LB, red dashed), together with the last-period-only lower bound (LP, green long-dashed).  \emph{Middle row}: Haile--Tamer bounds applied naively to each period; the HT lower bound (red dashed) crosses the upper bound, demonstrating misspecification in the sequential setting.  \emph{Bottom row}: Adaptive upper bound (blue dashed) and panel upper bound (green long-dashed) with our lower bound.  Black solid: true LogNormal$(\mu\!=\!1, \sigma\!=\!0.5)$ distribution.}
		\label{fig:sim}
	\end{figure}
	
	We now compare the six approaches from Section~\ref{sec:estimation_inference} for combining multiple order statistics in the upper bound $F_U$.  All methods share the same lower bound; only the upper bound computation differs.  We use the large-drop scenario ($\Delta = 0.05$, dropout $\omega_i \sim U(0.3, 1)$, $p_2 = \phi(Z_2) \cdot U(0.5, 1)$, 2{,}000 auctions across 200 replications) and vary the number of bidders: $N \in \{5, 10, 20\}$.  The theory in Section~\ref{subsec:finite_sample} predicts that the left-tail problem should worsen rapidly with~$N$. Table~\ref{tab:ub_comparison} reports the results.  We measure coverage (fraction of $v$ where $F_L(v) \leq F^0(v) \leq F_U(v)$), left-tail coverage (coverage restricted to $v$ where $F^0(v) < 0.1$), and invalidity (fraction of $v$ where $F_U(v) < F^0(v)$).
	
	\begin{table}[h!]
		\centering
		\caption{Comparison of upper bound estimation methods across $N$ (large-drop scenario, 200 MC iterations).}\label{tab:ub_comparison}
		\begin{tabular}{l r r r r r r r r r}
			\toprule
			& \multicolumn{3}{c}{Coverage (\%)} & \multicolumn{3}{c}{Left-tail cov.\ (\%)} & \multicolumn{3}{c}{Invalidity (\%)} \\
			\cmidrule(lr){2-4} \cmidrule(lr){5-7} \cmidrule(lr){8-10}
			Method & $N\!=\!5$ & $N\!=\!10$ & $N\!=\!20$ & $N\!=\!5$ & $N\!=\!10$ & $N\!=\!20$ & $N\!=\!5$ & $N\!=\!10$ & $N\!=\!20$ \\
			\midrule
			Baseline & 68.2 & 65.6 & 58.2 & 5.0 & 1.4 & 1.4 & 22.2 & 30.7 & 39.1 \\
			Na\"ive min          & 67.5 & 65.0 & 57.7 & 2.9 & 1.4 & 1.4 & 22.8 & 31.2 & 39.7 \\
			CLR                  & 67.6 & 65.3 & 58.0 & 2.9 & 1.4 & 1.4 & 22.7 & 30.9 & 39.3 \\
			Bernstein            & 33.9 & 30.3 & 25.7 & 2.9 & 1.4 & 1.4 & 56.2 & 66.1 & 71.6 \\
			Adaptive             & 86.1 & 90.5 & 90.7 & 97.6 & 100.0 & 100.0 & 5.7 & 8.0 & 9.0 \\
			IVW                  & 83.2 & 71.5 & 64.6 & 97.6 & 100.0 & 100.0 & 15.2 & 27.8 & 35.2 \\
			\bottomrule
		\end{tabular}
		\vspace{0.6ex}
		\begin{minipage}{\textwidth}
			\footnotesize \emph{Note:} Coverage is the fraction of the evaluation grid where $F_L(v) \leq F^0(v) \leq F_U(v)$. Left-tail coverage is restricted to $v$ where $F^0(v) < 0.1$.  Invalidity is the fraction where $F_U(v) < F^0(v)$.  All statistics averaged across 200 MC iterations of 2{,}000 auctions each.
		\end{minipage}
	\end{table}
	
	The left-tail problem worsens dramatically with~$N$: for the baseline method, left-tail coverage falls from the low single digits at $N = 5$ to near zero at $N = 20$, and the ``dead zone'' (the smallest $v$ at which $F_U(v) > 0$) approximately doubles with each doubling of $N$, consistent with the $T^{-1/N}$ rate in \citet{menzel2013large}. By contrast, the adaptive $j$-selection method dominates the other methods.  By excluding order statistics with degenerate empirical CDFs and using lower-ranked order statistics (which have support in the left tail), it achieves 100\% left-tail coverage at $N \geq 10$ with substantially lower invalidity than the baseline.  The adaptive method also has the best overall coverage among valid methods at every $N$. The CLR bias correction and Bernstein smoothing provide minimal improvement, or even degrade performance, relative to the baseline.  The CLR correction addresses extremum bias (the selection effect from taking the minimum), but this bias is negligible relative to the structural problem: the $j = 2$ empirical CDF is literally zero in the left tail, and no bias correction can recover information that is not in the data.  The Bernstein smoother pushes the empirical CDF positive near the boundary but overshoots, producing overly aggressive bounds with high invalidity. IVW shows left-tail coverage similar to the adaptive approach, but worse overall coverage and higher invalidity. The shape of the adaptive upper bound is displayed in the bottom panel of Figure~\ref{fig:sim}.
	
	Having established that the adaptive method dominates, we now examine whether proper inference preserves this advantage.  Table~\ref{tab:inference} reports the frequentist coverage of the 95\% confidence intervals described in Section~\ref{subsec:inference}, based on 200 Monte Carlo replications with 200 bootstrap draws each, comparing the baseline ($j = 2$ only) and adaptive approaches. The confidence bands substantially improve coverage: moving from estimated bounds to 95\% CIs raises coverage for the adaptive method by roughly 20 percentage points across all~$N$.  At $N = 20$, the adaptive method with both bootstrap and Imbens--Manski CIs achieves coverage at or above the nominal 95\% level with near-zero invalidity, demonstrating that the combination of adaptive $j$-selection and proper inference yields reliable finite-sample bounds even in high-competition settings where the baseline approach fails entirely in the tails.
	
	\begin{table}[h!]
		\centering
		\caption{Coverage of 95\% confidence intervals for $F^0$ (large-drop scenario, 200 MC $\times$ 200 bootstrap).}\label{tab:inference}
		\begin{tabular}{l l r r r r r r}
			\toprule
			& & \multicolumn{2}{c}{CI coverage (\%)} & \multicolumn{2}{c}{Left-tail CI cov.\ (\%)} & \multicolumn{2}{c}{CI invalidity (\%)} \\
			\cmidrule(lr){3-4} \cmidrule(lr){5-6} \cmidrule(lr){7-8}
			$N$ & Method & Boot & IM & Boot & IM & Boot & IM \\
			\midrule
			5 & Baseline ($j\!=\!2$) & 80.6 & 80.8 & 7.4 & 7.3 & 11.8 & 11.7 \\
			5 & Adaptive & 94.3 & 94.2 & 100.0 & 99.9 & 0.2 & 0.4 \\
			\addlinespace
			10 & Baseline ($j\!=\!2$) & 77.6 & 77.7 & 1.4 & 1.4 & 19.7 & 19.6 \\
			10 & Adaptive & 99.3 & 99.3 & 100.0 & 100.0 & 0.3 & 0.4 \\
			\addlinespace
			20 & Baseline ($j\!=\!2$) & 69.6 & 69.7 & 1.4 & 1.4 & 28.0 & 27.8 \\
			20 & Adaptive & 99.7 & 99.6 & 100.0 & 100.0 & 0.3 & 0.4 \\
			\bottomrule
		\end{tabular}
		\vspace{0.6ex}
		\begin{minipage}{\textwidth}
			\footnotesize \emph{Note:} Boot = bootstrap one-sided confidence band at 95\% level; IM = \citet{imbens2004confidence} confidence interval at 95\% level using analytical delta-method standard errors.  Coverage is the average fraction of grid points at which the CI contains the true~$F^0(v)$.  Left-tail coverage is restricted to $v$ where $F^0(v) < 0.1$.  CI invalidity is the fraction where the CI upper endpoint lies below the true~$F^0(v)$.  All statistics computed over 200 MC iterations with 200 bootstrap replications each.
		\end{minipage}
	\end{table}
		
	\subsection{Heterogeneous Number of Bidders}\label{subsec:hetN_sim}
	
	In practice (as in both of our empirical applications), the number of bidders tends to vary across auctions.  We validate the proposed moment-condition inversion approach by simulating 500 sequential auctions across 200 replications with $N_a \sim \text{DiscreteUniform}(2, 20)$, using the large-drop scenario.  Table~\ref{tab:hetN_sim} compares three methods for the bound estimates and reports 95\% CI coverage. The moment-condition inversion produces the tightest bounds (width~0.086, invalidity~0.8\%), outperforming the fixed median-$N$ proxy (width~0.097, invalidity~1.7\%).  The fixed-$N =$ min approach is essentially useless: using $N = 2$ for all auctions produces a bound with 87\% invalidity, because the beta shape $(1, 2)$ vastly overestimates the identifying power of high-$N_a$ auctions. With $N_a \sim \text{DiscreteUniform}(2, 20)$, the 95\% confidence intervals achieve 99.7\% overall coverage and 100\% left-tail coverage.  The strong performance reflects the presence of low-$N_a$ auctions in the sample: auctions with $N_a = 2$ or $3$ contribute informative beta weights at small~$F$, populating the left tail where the $j = 2$ ECDF would otherwise be degenerate.  This confirms that heterogeneity in~$N_a$ provides a natural solution to the left-tail dead zone problem when low-$N$ auctions are present in the data.
	
	\begin{table}[h!]
		\centering
		\caption{Bounds estimation with heterogeneous $N_a$}\label{tab:hetN_sim}
			\begin{tabular}{l r r}
			\toprule
			& Width & Invalidity (\%) \\
			\midrule
			\multicolumn{3}{l}{\emph{Point-estimate bounds}} \\
			\quad Moment-cond.\ inversion (het.\ $N$) & 0.086 & 0.8 \\
			\quad Fixed $N =$ median (11)              & 0.097 & 1.7 \\
			\quad Fixed $N =$ min (2)                   & 0.088 & 87.2 \\
			\addlinespace
			\multicolumn{3}{l}{\emph{95\% CI coverage (moment-cond.\ method)}} \\
			& CI coverage (\%) & Left-tail CI cov.\ (\%) \\
			\cmidrule(lr){2-3}
			\quad Bootstrap                            & 99.7 & 100.0 \\
			\quad Imbens--Manski                       & 99.7 & 100.0 \\
			\bottomrule
		\end{tabular}
		\vspace{0.6ex}
		\begin{minipage}{\textwidth}
			\footnotesize \emph{Note:} $N_a \sim \text{DiscreteUniform}(2, 20)$ across 500 sequential auctions per MC iteration.  The moment-condition inversion solves $T^{-1}\sum_{a=1}^T I_{F(v)}(N_a\!-\!1, 2) = T^{-1}\sum_{a=1}^T \mathbf{1}\{\eta_{2:N_a}^a \leq v\}$ by bisection, using each auction's individual~$N_a$.  CI coverage is the fraction of the grid where the CI contains the true~$F^0$.  Left-tail coverage is restricted to $v$ where $F^0(v) < 0.1$.
		\end{minipage}
	\end{table}
	
	Figure~\ref{fig:hetN_bounds} displays the mean bounds across 200 Monte Carlo iterations.  The moment-condition inversion and fixed-median bounds are visually close, while the fixed-min approach produces bounds that are too tight and lie below the true distribution over much of the support. In the figure, upper-bound curves are dashed and lower-bound curves are solid.
	
	\begin{figure}[h!]
		\centering
		\includegraphics[width=\textwidth]{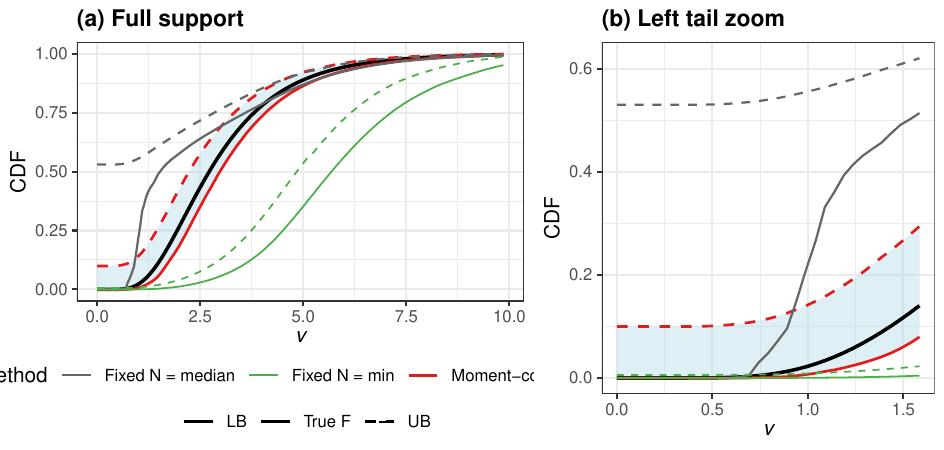}
		\caption{Mean bounds on the valuation distribution with heterogeneous $N_a$ (200 MC iterations, 500 auctions each with $N_a \sim \text{DiscreteUniform}(2,20)$).  Black solid: true LogNormal$(\mu\!=\!1, \sigma\!=\!0.5)$ distribution.  The moment-condition inversion and fixed $N =$ median produce visually close bounds, while fixed $N =$ min is invalid over much of the support.  Upper-bound curves are dashed and lower-bound curves are solid.}\label{fig:hetN_bounds}
	\end{figure}
	
	\begin{remark}[Adaptive $j$-selection vs.\ moment-condition inversion]\label{rem:adaptive_vs_moment}
	The simulation results reveal a complementarity between two approaches for improving left-tail identification. \emph{Adaptive $j$-selection} populates the left tail by combining multiple order statistics ($j = 2, 3, \ldots$) whose distributions $G_{j:N}$ extend progressively further into the left tail.  It works best with fixed~$N$ and moderate to large competition (Table~\ref{tab:ub_comparison}: at $N = 20$, 100\% left-tail coverage with lower invalidity than the baseline), but the differing beta shapes across~$j$ can cause finite-sample crossing when~$N$ is small. \emph{Moment-condition inversion} achieves left-tail identification through a different mechanism: variation in~$N_a$ across auctions, with low-$N_a$ auctions contributing informative beta weights at small~$F$.  Because it uses a single order statistic ($j = 2$), crossing is eliminated by construction, but the method requires substantial heterogeneity in~$N_a$. In practice, the choice is dictated by the data.  When~$N$ is fixed and multiple order statistics are observed, the adaptive method is preferred.  When~$N_a$ varies substantially across auctions, the moment-condition inversion is preferred because it avoids crossing while achieving comparable identification power.
	\end{remark}
	
	\section{Empirical Applications}\label{sec:empirical}
	
	\subsection{Korean Wholesale Used Car Auctions}\label{subsec:korean}
	
	We use data from a wholesale used car auction house located in Suwon, South Korea, which opened in May 2000 as the first fully computerized automobile auction facility in the country \citep{il2014nonparametric, roberts2013unobserved}.  The auction house held weekly sales from October 2001 through December 2002, our sample period. Between 250 and 380 registered dealer members were eligible to bid, with roughly half attending any given weekly auction.  The pool of registered dealers remained stable throughout the sample period.  Sellers, predominantly individual car owners, rental companies, and corporate fleet operators, consigned vehicles and paid a listing fee of approximately \$50.  The auction house charged a 2.2\% commission on the selling price to both buyer and seller upon sale. A few days before each weekly auction, potential buyers received a catalog with detailed information about every car to be sold, including make, model, year, mileage, engine size, transmission type, fuel type, and accident history.  Buyers could also physically inspect the cars 2--3 hours before the auction.
	
	On each auction day, several hundred to over a thousand cars were offered sequentially (median 733 consignments per day) to a stable pool of 150--200 active dealers. Approximately half of consigned cars sold at any given auction and unsold cars were typically reconsigned for a future auction day. The auction followed an open ascending-price format similar to the English button auction \citep{milgrom1982theory}.  Each car had a reserve price set by the seller. The auctioneer announced the opening price and the price rose in fixed increments of 30{,}000 KRW (approximately \$20) every 3 seconds.  Buyers signaled willingness to buy at the current price by pressing a button on an electronic device.  The auction house displayed only the current price and a traffic-light indicator: green ($\geq 3$ active bidders), yellow (2 active bidders), and red (1 active bidder).  Bidder identities were not revealed.  Importantly, reentry was permitted: a buyer could drop out temporarily and rejoin later if the price remained competitive, distinguishing this format from the strict button auction where exit is irrevocable.  When only one buyer remained active, the auction ended and the winner paid the current price. A car that sold was typically paid for and removed within a few days.  
	
	The raw data contain 160 auction days, each comprising a series of sequential auctions.  We observe the top three bids and corresponding bidder identifiers for each auction, along with vehicle characteristics and auction outcomes.  Of 100{,}858 total auction records, 53{,}241 resulted in a sale, of which 48{,}753 have the top-two bids recorded and enter our analysis.  To maintain the IPV assumption, we estimate bounds separately within five major vehicle segments: automatic sedans (gasoline), manual sedans (gasoline), manual non-sedans (gasoline), automatic non-sedans (gasoline), and manual non-sedans (diesel).  These five categories cover 43{,}165 of the 48{,}753 usable sold auctions (88.5\%). Within each category, we construct all consecutive auction pairs. 
	
	Estimation proceeds in two steps.  In non-terminal periods, bids are strategically discounted by the continuation value: $b_{i,k} = \theta_i\phi(Z_k) - \text{CV}_k(\theta_i)$.  This breaks the multiplicative structure needed for hedonic identification, since $\log b_{i,k} \neq \log\theta_i + \log\phi(Z_k)$.  Only in the last period, where $\text{CV}_K = 0$, does the log-additive decomposition hold exactly.  We therefore use last-period bids to estimate $\phi(Z_k)$ via hedonic regression:
	\begin{equation}\label{eq:hedonic}
		\log b_{i,K}^m = \alpha_i + Z_K^{m\prime}\beta + \varepsilon_{iK}^m,
	\end{equation}
	where the superscript $m$ indexes the sequential auction series and $Z_K^m$ collects observed characteristics of the item in the last period of series $m$.  Bidder fixed effects absorb $\log\theta_i$, and the coefficients on car characteristics identify $\phi(\cdot)$.  Standard errors are clustered at the auction-day level. Taking $\hat{\phi}(Z_k)$ as given, we estimate the upper and lower bounds on $F(\cdot)$. Table~\ref{tab:hedonic} reports the hedonic regression estimates. The coefficients have expected signs: older cars sell for less, larger engine size and higher condition ratings increase value, and the point estimate for imported cars is positive.
	
	\begin{table}[h]
		\centering
		\caption{Hedonic Regression for the Common Value Component}\label{tab:hedonic}
		\begin{tabular}{lclc} \toprule
			Variable & Coefficient & Variable & Coefficient \\ \midrule
			Car age (years)          & $-0.328^{***}$   & Automatic transmission  & $0.303^{***}$ \\
			& $(0.012)$        &                         & $(0.035)$     \\[3pt]
			$\log(\text{Mileage})$          & $0.023$          & Imported vehicle        & $0.219$       \\
			& $(0.033)$        &                         & $(0.412)$     \\[3pt]
			$\log(\text{Engine size})$      & $0.997^{***}$    & Hyundai                 & $0.344^{***}$ \\
			& $(0.074)$        &                         & $(0.057)$     \\[3pt]
			Condition rating (1--6)         & $0.063^{***}$    & Kia                     & $0.090$       \\
			& $(0.015)$        &                         & $(0.056)$     \\[3pt]
			Gasoline                        & $0.191^{**}$     & Sedan                   & $-0.112^{*}$  \\
			& $(0.073)$        &                         & $(0.058)$     \\ \midrule
			Bidder FE         & Yes            &                         &              \\
			Observations      & 1{,}782        &                         &              \\
			Adjusted $R^2$    & 0.799          &                         &              \\ \bottomrule
		\end{tabular}
		\vspace{2mm}
		
		\begin{minipage}{\textwidth}
			\footnotesize \emph{Note:} Dependent variable: $\log(\text{bid}_{i,K}^m)$, where $K$ denotes the last period within each market. The omitted brand category comprises other domestic manufacturers (Daewoo, Ssangyong, Samsung, and Asia).  The family-use indicator is not reported because it is collinear with the retained vehicle categories. Standard errors in parentheses are clustered at the auction-day level.  $^{***}p<0.01$, $^{**}p<0.05$, $^{*}p<0.1$.
		\end{minipage}
	\end{table}
	
	The total number of potential bidders is not directly observed. We observe only the top three bids per auction, while the pool of registered dealers may not all be active in every market.  We proxy for $N_a$ in each market by counting the number of unique dealers who won at least one auction within that market.  The estimated~$N_a$ varies substantially: the median ranges from~28 (manual non-sedan diesel and automatic non-sedan gasoline) to~57 (automatic sedan gasoline), reflecting different competition intensities and popularities across vehicle segments.  We use the moment-condition inversion method to estimate the bounds because the data exhibit heterogeneous~$N_a$. Figure~\ref{fig:korean_bounds} displays the estimated upper bound and three lower-bound specifications. The baseline lower bound uses the adjacent-opportunity specification, with $C_k=\hat\phi(Z_{k+1})$ and $m_k=0$, combined with the terminal-period inequality in Lemma~\ref{lem:ub_last} without imposing the global monotonicity condition.\footnote{Under this specification, 15{,}817 sequential pairs that satisfy the filter $\hat\phi(Z_k)>\hat\phi(Z_{k+1})$ are retained to construct the lower bound.} This specification treats the next car in the same auction-day category as the relevant substitute opportunity in the dealer's marginal procurement problem. The interpretation is natural in the Korean used-car auction setting because vehicles are sold in a posted sequence, dealers observe a dense same-day stream of similar alternatives, and the immediately following auction provides a transparent proxy for the continuation opportunity. Setting $m_k=0$ bounds continuation value by the gross value of the next opportunity without using the opening price. Thus all descending adjacent pairs with positive finite statistics contribute. The baseline bounds have mean widths of 0.029--0.107 across categories. 
	
	Figure~\ref{fig:korean_bounds} also reports the sharper net specification (subtracting the next-period opening price), retaining pairs that satisfy the normalized positivity filter $b^{1:n_k}/\hat\phi(Z_k)>p_{k+1}/\hat\phi(Z_{k+1})$, as well as the terminal-only lower bound.  The sharper net bounds have mean widths of 0.029--0.111. They are not uniformly tighter because changing the statistic also changes the retained sample and the moment-condition inversion weights.  Appendix~\ref{app:korean_monotonicity} conducts robustness checks that instead use the conservative full-future index $C_k=\max_{t\geq k+1}\hat\phi(Z_t)$ and tail-to-terminal subsequences that satisfy Condition~\ref{cond:descending}. The resulting loss of identifying power is concentrated in the left tail, while the bounds are largely unchanged over the rest of the support. The 95\% Imbens--Manski confidence intervals have mean widths of 0.048--0.136.
	
	Dealers in this market are wholesale buyers who often purchase multiple cars per auction day.  We interpret the unit-demand assumption at the level of a marginal procurement opportunity rather than at the level of the dealer's entire daily business. A dealer bidding on a given vehicle decides whether to satisfy a current inventory need now or preserve the option to purchase a similar vehicle later, possibly at a lower price.  If the dealer wins, that marginal need is satisfied; subsequent participation can be interpreted as a new procurement spell, or equivalently as a new bidder in the model. Under this renewal interpretation, observed multi-car purchases by the same dealer are consistent with unit demand within each decision problem.  As a robustness check for the more literal fixed multi-unit-demand interpretation, we compute the terminal-only lower bound. This bound requires no sequential opportunity-cost restriction and is therefore valid even if dealers have several independent purchase needs within the same auction day.  The resulting terminal-only bounds have mean widths of 0.029--0.129 and are largely consistent with the sequential bounds, though wider in the lower tail.
	
	\begin{figure}[p]
		\centering
		\includegraphics[width=0.94\textwidth]{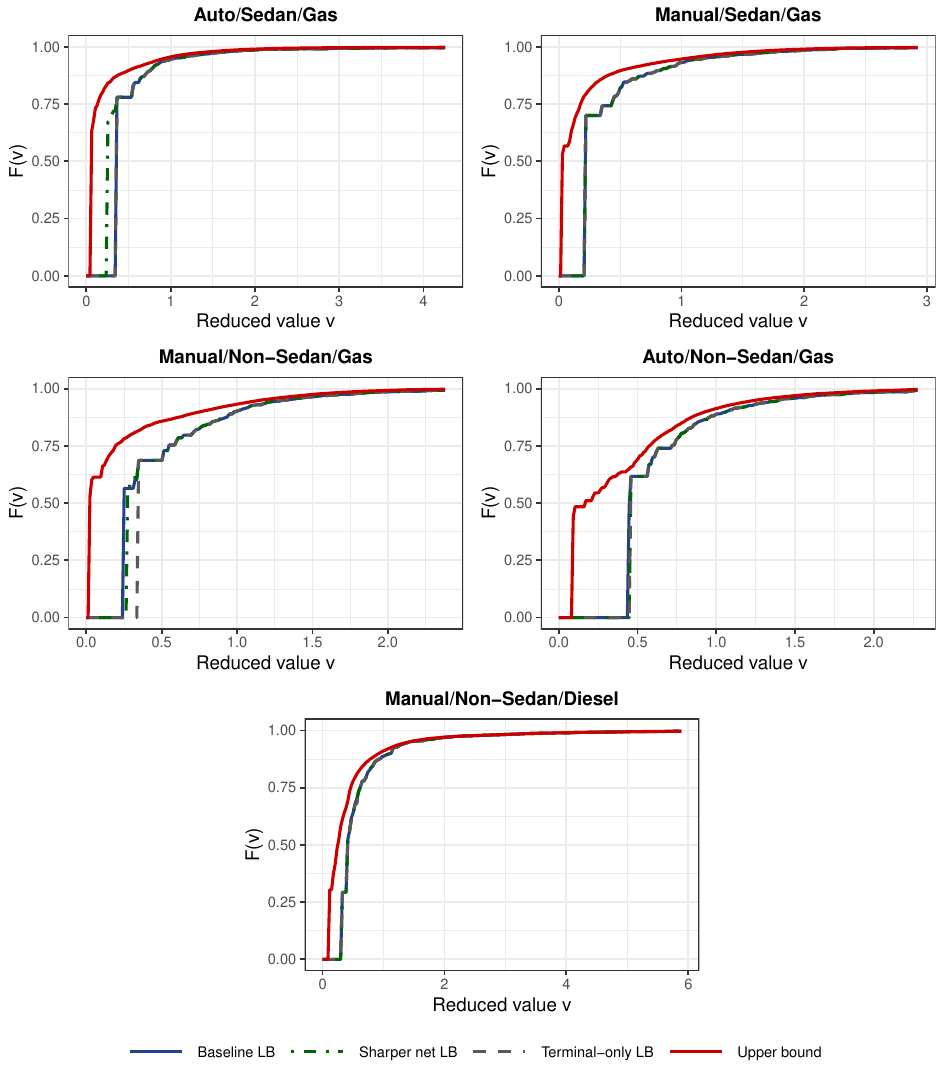}
		\caption{Estimated bounds on the valuation distribution by vehicle category (Korean used car data).  Red solid: no-overbidding upper bound.  Blue solid: baseline lower bound, combined with the terminal-period moment.  Dark green dot-dashed: sharper net lower bound, combined with the terminal-period moment.  Gray dashed: terminal-only lower bound. All bounds use each market's estimated~$N_a$ in the moment-condition inversion.}\label{fig:korean_bounds}
	\end{figure}
	
	\subsection{Cars and Bids Online Auctions}\label{subsec:carsandbids}
	
	We apply our bounds approach to a second dataset: online ascending auctions from the Cars and Bids platform.  This setting differs fundamentally from the Korean data in one key respect: there is no terminal auction.  Individual listings run for seven days, new auctions open continuously, and a bidder tracking a particular vehicle segment faces an ongoing stream of opportunities with no defined endpoint.  In this environment, the static HT framework cannot produce a valid lower bound on~$F$ because waiting for a future listing is always a rational alternative.  Our sequential approach fills this gap by exploiting the opportunity cost structure across consecutive auctions to provide a lower bound.  The Cars and Bids data also offer two practical advantages: (i)~the number of unique bidders~$N$ is known exactly from complete bid histories, and (ii)~the moderate-competition environment tests the finite-sample methods developed in Section~\ref{sec:estimation_inference}.
	
	Cars and Bids is a U.S.-based online auction platform specializing in enthusiast vehicles.  We use data from January 2024 through July 2025, comprising 10{,}811 auctions after filtering to auctions with at least two bidders and valid geographic information (mapped to U.S. Metropolitan Statistical Areas).  Unlike the Korean data, complete bid histories are observed: every bid amount, bidder identity, and timestamp.  The median auction attracts about 14~unique bidders (range 2--44), each placing approximately two bids on average, with a median winning bid of \$18{,}750 (mean~\$29{,}631).  We group vehicles into five market segments using $k$-prototypes clustering, a mixed-data analogue of $k$-means, on year, make, model, body style, and transmission type, imposing the IPV assumption within clusters rather than across the full dataset.
	
	We construct bidder-specific sequences from the complete bid histories. For each bidder, we first record the auctions in which the bidder placed at least one bid, merge those bidder-auction observations with the auction's MSA, end month, end date, and estimated common value, and then order the auctions in which the bidder participated within each MSA-month cell. Consecutive observations in this bidder-specific ordering define a potential current and future opportunity pair. We retain pairs for which the current auction has a higher estimated common value than the bidder's next observed auction.  This yields 4{,}137 sequential pairs from 1{,}388 unique bidders, roughly half of all 8{,}489 consecutive pairs. 
	
	We estimate the common value function $\phi(Z_k)$ via hedonic regression of $\log(\text{highest bid})$ on vehicle characteristics, using 10{,}811 auctions with make and model fixed effects.  The car-characteristics model includes standard attributes such as year, mileage, engine, transmission, drivetrain, and body style. We use this car-characteristics-only model to construct $\hat\phi(Z)$ because $\phi(Z)$ is intended to capture predetermined object characteristics. For comparison, we also estimate a full model that adds text-based features from the listing: the length of seller highlights, service history, modifications, and seller notes, as well as the number of images and videos. We also incorporate comment sentiment features derived from a lexicon-based analysis of auction comments.\footnote{We use the AFINN sentiment lexicon \citep{nielsen2011new}, which assigns integer scores from $-5$ (most negative) to $+5$ (most positive) to approximately 2{,}500 common English words (e.g., ``excellent'' $= +3$, ``broken'' $= -3$).  For each auction, we tokenize all user comments posted on the listing page into individual words, match them against the AFINN lexicon, and compute the \emph{weighted sentiment} as the average AFINN score across all matched words.  This score measures the overall tone of community engagement: a high value indicates enthusiastic, positive commentary, while a low value suggests concerns or criticism.  We also include the count of matched sentiment-bearing words as a separate regressor, capturing the volume of substantive discussion.} These variables improve price prediction but may partly reflect endogenous attention, information revealed during the auction, or bidder sentiment, so they are not used in the baseline common-value index.  Table~\ref{tab:cb_hedonic} reports key coefficients from the hedonic regression.  The car-characteristics-only model used for $\hat\phi$ achieves an adjusted~$R^2$ of~0.819; including text and sentiment features raises it to~0.845.
	
	\begin{table}[h!]\centering
		\caption{Hedonic Regression: Cars and Bids Common Value Function}\label{tab:cb_hedonic}
		\begin{tabular}{lcc} \hline
			& Car-only model ($\hat\phi$) & Full model \\ \hline
			Year                      & $0.067^{***}$ & $0.063^{***}$ \\
			& (0.002) & (0.002) \\
			$\log(\text{Mileage})$    & $-0.148^{***}$ & $-0.142^{***}$ \\
			& (0.005) & (0.005) \\
			Sentiment (weighted)      & --- & $0.085^{***}$ \\
			&     & (0.020) \\
			Listing length (words)    & --- & $0.032^{***}$ \\
			&     & (0.004) \\
			Number of images          & --- & $0.003^{***}$ \\
			&     & (0.001) \\[4pt]
			Make/Model FE             & Yes & Yes \\
			Drivetrain/transmission/body-style FE & Yes & Yes \\
			Observations              & 10{,}811 & 10{,}811 \\
			Adjusted $R^2$            & 0.819 & 0.845 \\ \hline
		\end{tabular}
		\vspace{0.6ex}
		\begin{minipage}{\textwidth}
			\footnotesize \emph{Note:} Dependent variable is $\log(\text{highest bid})$.  Standard errors in parentheses.  $^{***}p<0.01$.  The car-characteristics-only model is used for $\hat\phi(Z_k)$.  Make and model fixed effects absorb brand-level price differences (643 make/model combinations).
		\end{minipage}
	\end{table}

    In the Korean application, the hedonic regression uses only last-period bids and includes bidder fixed effects, because the continuation value~$\text{CV}_k(\theta_i)$ in non-terminal periods breaks the log-additive structure needed for identification.  In the Cars and Bids setting, this approach is infeasible for two reasons.  First, each Cars and Bids listing is an independent 7-day online auction; the sequential structure is imposed \emph{ex post} by tracking individual bidders across separate auctions, so there is no platform-defined terminal period that applies to all bidders simultaneously.  Second, unlike the stable pool of 380 registered dealers in the Korean data, Cars and Bids bidders are highly heterogeneous and many participate in only one or two auctions, making bidder fixed effects impractical. Instead, we estimate the common-value index $\phi(Z)$ from the auction-level transaction price in the full cross-section of Cars and Bids auctions. The transaction price is the most informative object-level price signal available in the absence of terminal auctions or a stable bidder panel. Sequential incentives may affect transaction prices through bidder-specific continuation values; the maintained requirement is that, after conditioning on rich vehicle characteristics and make/model fixed effects, these residual strategic components are not systematically related to the observed characteristics used to construct the common-value index. The sequential structure is then accounted for in the second step through the opportunity-cost bounds, not by treating the auction as static.
    
	With complete bid histories, we observe the exact number of unique bidders $N_a$ in each auction~$a$.  We use the moment-condition inversion method with each auction's exact~$N_a$.  For the lower bound, we implement the myopic continuation specification $C_k=\hat\phi(Z_{k+1})$ with zero opening prices. Because the Cars and Bids market is effectively infinite horizon, bidders only know the listings that are currently active but not the future composition of cars that will arrive on the platform.  We therefore treat the bidder's next observed auction in the local sequence as the relevant continuation option.  The resulting statistic is $h_{1:n_k} = b^{1:n_k}/(\hat\phi(Z_k) - \hat\phi(Z_{k+1}))$.  We use $j = 2$ for the upper bound on~$F$ (no-overbidding). Table~\ref{tab:cb_bounds} reports the bound width for each vehicle cluster.  The bounds have mean widths of 11.7--18.1 percentage points across clusters.  The 95\% Imbens--Manski confidence intervals are 15--21 percentage points wide, properly accounting for estimation uncertainty. 
	
	\begin{table}[h!]\centering
		\caption{Bounds on the valuation distribution by vehicle cluster (Cars and Bids data).}\label{tab:cb_bounds}
		\begin{tabular}{l r r r r r} \toprule
			& Cluster~1 & Cluster~2 & Cluster~3 & Cluster~4 & Cluster~5 \\ \midrule
			Sequential pairs      & 993   & 696   & 883   & 1{,}032 & 533   \\
			Unique auctions       & 1{,}352 & 945 & 1{,}250 & 1{,}491 & 780 \\
			Median $N$            & 14    & 14    & 14    & 14    & 14    \\
			\addlinespace
			\multicolumn{6}{l}{\emph{Bound width ($j=2$ UB / sequential LB)}} \\
			\quad Mean width       & 0.177 & 0.117 & 0.144 & 0.181 & 0.166 \\
			\quad Median width     & 0.169 & 0.110 & 0.132 & 0.152 & 0.145 \\
			\quad 95\% IM CI (mean) & 0.212 & 0.152 & 0.181 & 0.209 & 0.198 \\
			\bottomrule
		\end{tabular}
		\vspace{0.6ex}
		
		\begin{minipage}{\textwidth}
			\footnotesize \emph{Note:} Width is the average of $F_U(v) - F_L(v)$ across the evaluation grid.  The 95\% Imbens--Manski CI uses analytical delta-method standard errors.
		\end{minipage}
	\end{table}
	
	\begin{figure}[p]
		\centering
		\includegraphics[width=\textwidth]{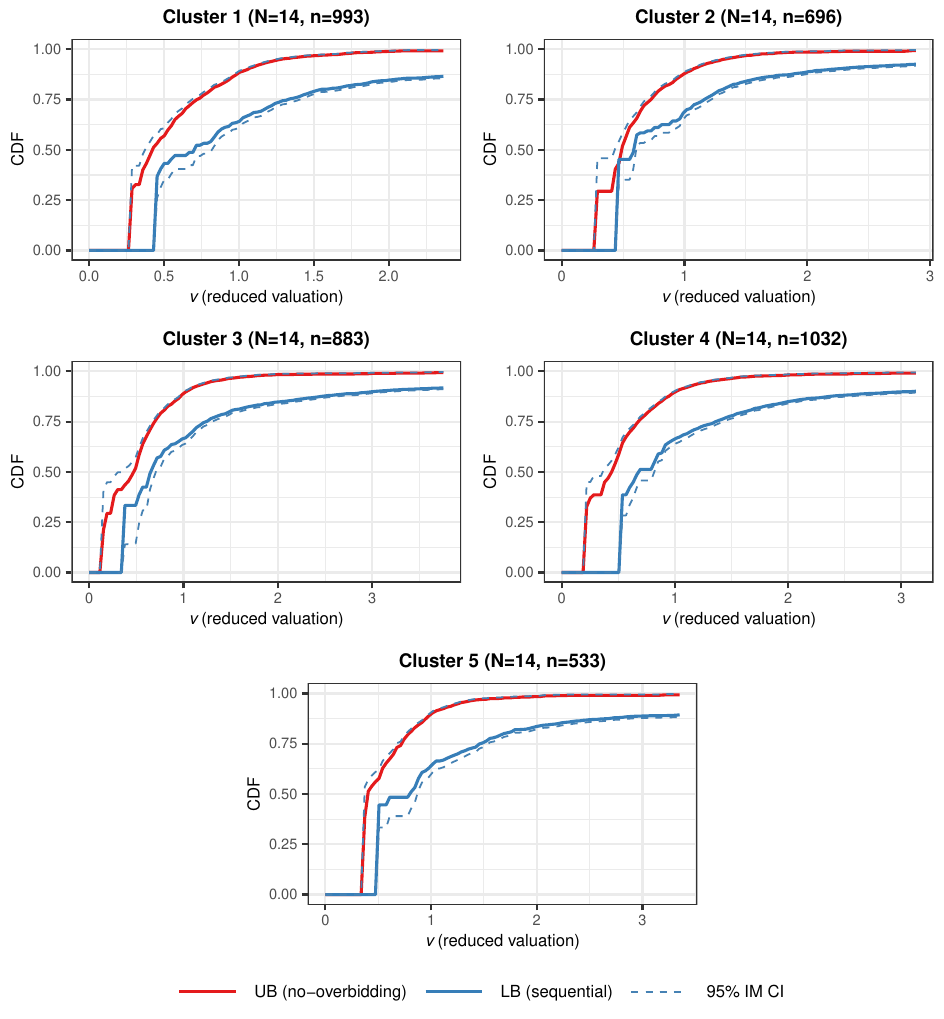}
		\caption{Estimated bounds on the valuation distribution by vehicle cluster (Cars and Bids data).  Each panel shows the upper bound on~$F$ (red solid) and the lower bound on~$F$ (blue solid), with 95\% Imbens--Manski confidence intervals (steelblue dashed).}\label{fig:cb_bounds}
	\end{figure}
	
	Since complete bid histories are available, we can also construct a \emph{panel upper bound} from individual bidders' maximum reduced bids: $F(v) \leq G^{\max}(v) = \widehat{\Pr}(\eta_i^{\max} \leq v)$, using 22{,}000--49{,}000 unique bidders per cluster.  This bound requires no beta inversion and no assumption on~$N_a$.  Comparing the panel upper bound to the sequential lower bound reveals that the baseline lower bound (computed with actual~$N_a$) exceeds the panel upper bound at 20--30 grid points per cluster.  Under the maintained assumptions, both bounds should be valid, so this crossing suggests that not all $N_a$ bidders are effectively competing in the auction. Casual bidders who place low bids may have systematically lower valuations than serious competitors, inflating the effective~$N$ in the beta inversion.
	
	The panel bound diagnostic motivates an investigation of effective competition.  We estimate the ``effective $N$'' using four complementary methods in Appendix~\ref{app:effective_N}: (i)~counting bidders whose maximum bid exceeds 50\% of the winning bid (median~$N_{\text{eff}} = 8$), (ii)~counting active final-stage bidders ($N_{\text{eff}} \approx 5$, limited by timestamp coarseness), (iii)~estimating~$N$ from order statistic spacing ($N_{\text{eff}} = 6$--$7$), and (iv)~counting bidders within the top 90\% of the bid range ($N_{\text{eff}} = 10$--$11$).  All four methods broadly point to $N_{\text{eff}} \approx 5$--$10$, suggesting that roughly half of the unique bidders are casual participants.  Re-estimating the lower bound with~$N_{\text{eff}}$ from Methods~1 and~3 nearly eliminates the crossing with the panel upper bound, while the qualitative shape of the bounds is preserved.  The bound width increases from 11.7--18.1 percentage points (baseline) to 15--29 percentage points ($N_{\text{eff}}$-adjusted), reflecting the reduced information from fewer effective competitors.  The baseline results using actual~$N_a$ should therefore be interpreted as an aggressive estimate, with the $N_{\text{eff}}$-adjusted bounds providing a conservative alternative.	
	
	\section{Counterfactual Analyses}\label{sec:counterfactual}
	
	A key practical question is whether bounds on the valuation distribution can support meaningful policy analysis.  This section answers in the affirmative through three exercises that translate the bounds $[F_L, F_U]$ into bounds on quantities of direct interest to sellers, platform designers, and regulators.  The first quantifies how expected revenue depends on the continuation probability of future auctions in the Korean wholesale market: this is the empirical counterpart to the paper's central theoretical mechanism, the option to wait.  The second asks how Cars and Bids seller revenue would change if the platform converted casual browsers into serious competitors, providing a price tag on bidder-engagement mechanisms such as qualifying deposits or activity-based filters.  The third investigates the reserve-price envelope to identify a robust set of reserve choices that do not require committing to any particular $F$ within the bounds. Because each counterfactual reduces to an order-statistic functional of $F$, the bounds on $F$ map directly into bounds on the policy quantity, with no equilibrium or distributional assumptions beyond those already maintained.
	
	The key building block is the expected order statistic: for $N$ i.i.d.\ draws from $F$, the $j$th highest order statistic has CDF $G_{j:N}(v;F) = I_{F(v)}(N\!-\!j\!+\!1, j)$.  Since $I_z(a,b)$ is increasing in~$z$, $F_L \leq F \leq F_U$ implies
	\begin{equation}\label{eq:E_order_stat_bounds}
		E_{F_U}[\theta^{j:N}] \leq E_{F}[\theta^{j:N}] \leq E_{F_L}[\theta^{j:N}],
	\end{equation}
	where $E_F[\theta^{j:N}] = \underline{\theta} + \int_{\underline{\theta}}^{\bar{\theta}} [1 - G_{j:N}(v;F)]\, dv$.  For the winner rent, we use the identity $E_F[\theta^{1:N} - \theta^{2:N}] = \int_{\underline{\theta}}^{\bar{\theta}} N F(v)^{N-1}(1 - F(v))\, dv$.  Since the integrand depends only on $F(v)$ pointwise, bounds follow from minimizing and maximizing the integrand over $F(v) \in [F_L(v), F_U(v)]$ at each~$v$ (which has a closed form since $N F^{N-1}(1-F)$ is unimodal in~$F$ with maximum at $F^* = (N-1)/N$).
	
	\subsection{Future-Auction Uncertainty in the Korean Market}
	
	Sequential auction markets are routinely subject to disruptions that threaten the continuity of supply: weather and logistics shocks at wholesale yards, regulatory pauses, dealer strikes, and shifts in consignment volume can all reduce or eliminate the next sale opportunity.  These disruptions are not benign supply-side events; by changing what bidders expect of future auctions, they alter the competitive dynamics of the current sale.  In our framework, the option value of waiting is what depresses the current-period transaction price relative to a static benchmark, and when the option to wait is curtailed, they bid more aggressively today, raising current revenue at the cost of foregone future revenue.  Quantifying this trade-off is essential for any platform that wishes to evaluate the revenue cost of a delay, the value of guaranteeing a future sale schedule, or the design of cancellation and rescheduling policies.
	
	To formalize the exercise, we perturb the continuation probability $\tau \in [0,1]$: with this probability, the next auction in the observed sequence occurs, and with probability $1-\tau$, the future opportunity disappears.  The case $\tau = 1$ corresponds to the certainty of a next auction (bidders fully internalize the option to wait), while $\tau = 0$ corresponds to a static auction with no continuation value.  We implement the counterfactual on Korean auction-day pairs. For each consecutive pair of auctions~$a$, with estimated common values $\hat\phi_{1a} > \hat\phi_{2a}$ and market competition~$N_a$, the two-period benchmark model implies
	\begin{align}
		E[P_{1a}(\tau)] = (\hat\phi_{1a} - \tau\hat\phi_{2a}) E[\theta^{2:N_a}]
		+ \tau\hat\phi_{2a} E[\theta^{3:N_a}], \quad
		E[P_{2a}(\tau)] = \tau\hat\phi_{2a} E[\theta^{3:N_a}].
	\end{align}
	The counterfactual therefore does not require simulating bid paths under a new dynamic equilibrium: the equilibrium price formula reduces the exercise to bounded order-statistic expectations, evaluated at the category-specific estimates of $F_L$ and~$F_U$.  We average the resulting bounds across observed descending pairs with $N_a \geq 3$.
	
	A particularly clean object that emerges from this exercise is the \emph{waiting discount}: the difference between first-period revenue under no future auction ($\tau = 0$) and first-period revenue when the future auction occurs with probability~$\tau$.  Subtracting $E[P_1(\tau)] = (\hat\phi_1 - \tau\hat\phi_2) E[\theta^{2:N}] + \tau\hat\phi_2 E[\theta^{3:N}]$ from $E[P_1^{\text{static}}] = \hat\phi_1 E[\theta^{2:N}]$ gives, after cancellation,
	\begin{equation}\label{eq:waiting_discount}
		E[P_1^{\text{static}}] - E[P_1(\tau)] = \tau \hat\phi_2 \big(E[\theta^{2:N}] - E[\theta^{3:N}]\big).
	\end{equation}
	This decomposition isolates the option-value mechanism: the discount is exactly proportional to the gap between the second- and third-highest order statistics, scaled by $\tau\hat\phi_2$.  Because both revenue components and the waiting discount are affine in~$\tau$, Table~\ref{tab:cf_korean_continuation} reports the full counterfactual content on a compact grid of continuation probabilities; values between grid points are obtained by linear interpolation. Moving from $\tau = 0$ to $\tau = 1$ lowers the first-period lower envelope by 8.7--10.9\% and the upper envelope by 8.1--10.0\% across the five categories, while total expected revenue over the two-auction window rises by 12.9--29.2\% on the lower envelope and 15.2--30.3\% on the upper envelope. The waiting discount at $\tau=1$ ranges from 4.3--55.5 on the lower bound to 47.7--94.4 on the upper bound, in units of 10{,}000 KRW, with the largest discount in the Manual Non-Sedan Diesel segment where the order-statistic gap is widest.
	
	\begin{table}[t!]
		\centering
		\caption{Expected revenues and waiting discount with future uncertainty.}\label{tab:cf_korean_continuation}
		\resizebox{\textwidth}{!}{%
		\begin{tabular}{lccccc} \toprule
				Category & $\tau$ & $E[P_1(\tau)]$ & $E[P_2(\tau)]$ & $E[P_1(\tau)+P_2(\tau)]$ & Waiting discount \\ \midrule
				Auto/Sedan/Gas & 0.00 & [508.9, 606.6] & [0.0, 0.0] & [508.9, 606.6] & [0.0, 0.0] \\
				 & 0.25 & [495.5, 591.4] & [39.3, 47.7] & [534.8, 639.1] & [5.0, 23.6] \\
				 & 0.50 & [482.1, 576.2] & [78.6, 95.4] & [560.7, 671.6] & [10.0, 47.2] \\
				 & 0.75 & [468.7, 561.0] & [118.0, 143.2] & [586.7, 704.1] & [15.0, 70.8] \\
				 & 1.00 & [455.3, 545.8] & [157.3, 190.9] & [612.6, 736.6] & [19.9, 94.4] \\
				\addlinespace
				Manual/Sedan/Gas & 0.00 & [257.7, 304.8] & [0.0, 0.0] & [257.7, 304.8] & [0.0, 0.0] \\
				 & 0.25 & [250.7, 298.0] & [19.2, 24.3] & [269.9, 322.2] & [2.0, 11.9] \\
				 & 0.50 & [243.6, 291.1] & [38.4, 48.6] & [282.0, 339.7] & [3.9, 23.9] \\
				 & 0.75 & [236.6, 284.3] & [57.6, 72.9] & [294.2, 357.1] & [5.9, 35.8] \\
				 & 1.00 & [229.5, 277.4] & [76.8, 97.1] & [306.4, 374.5] & [7.9, 47.7] \\
				\addlinespace
				Manual/Non-Sedan/Gas & 0.00 & [312.7, 369.5] & [0.0, 0.0] & [312.7, 369.5] & [0.0, 0.0] \\
				 & 0.25 & [304.7, 362.0] & [27.4, 34.4] & [332.1, 396.4] & [1.1, 14.5] \\
				 & 0.50 & [296.7, 354.6] & [54.8, 68.8] & [351.4, 423.4] & [2.2, 29.0] \\
				 & 0.75 & [288.6, 347.1] & [82.1, 103.2] & [370.8, 450.3] & [3.2, 43.5] \\
				 & 1.00 & [280.6, 339.6] & [109.5, 137.6] & [390.1, 477.2] & [4.3, 58.0] \\
				\addlinespace
				Auto/Non-Sedan/Gas & 0.00 & [428.3, 479.4] & [0.0, 0.0] & [428.3, 479.4] & [0.0, 0.0] \\
				 & 0.25 & [418.9, 469.6] & [40.6, 46.1] & [459.6, 515.7] & [3.8, 15.3] \\
				 & 0.50 & [409.6, 459.8] & [81.2, 92.2] & [490.8, 552.0] & [7.7, 30.6] \\
				 & 0.75 & [400.3, 449.9] & [121.9, 138.3] & [522.1, 588.3] & [11.5, 46.0] \\
				 & 1.00 & [391.0, 440.1] & [162.5, 184.4] & [553.4, 624.5] & [15.3, 61.3] \\
				\addlinespace
				Manual/Non-Sedan/Diesel & 0.00 & [679.5, 718.0] & [0.0, 0.0] & [679.5, 718.0] & [0.0, 0.0] \\
				 & 0.25 & [661.3, 700.9] & [40.1, 44.5] & [701.5, 745.3] & [13.9, 21.5] \\
				 & 0.50 & [643.1, 683.7] & [80.3, 89.0] & [723.4, 772.6] & [27.7, 43.1] \\
				 & 0.75 & [624.9, 666.5] & [120.4, 133.4] & [745.3, 799.9] & [41.6, 64.6] \\
				 & 1.00 & [606.7, 649.3] & [160.5, 177.9] & [767.2, 827.2] & [55.5, 86.2] \\ \bottomrule
			\end{tabular}%
		}
		\vspace{0.6ex}
		\begin{minipage}{\textwidth}
			\footnotesize \emph{Note:} Entries are bounds $[\text{lower},\text{upper}]$ in units of 10{,}000 KRW.
		\end{minipage}
	\end{table}
	
	The option-value effect places a lower bound on the cost of supply disruptions: in the Korean market, an unanticipated cancellation of the next sale would raise current-period revenue by about 9--12\%, but the platform loses the entire next-period revenue stream, yielding a net loss of 11--23\% on the two-auction window. Sellers and platforms can therefore use the bounds to value insurance, hedging contracts, or guaranteed-schedule mechanisms that reduce~$1-\tau$. The asymmetry between the period-1 gain and the total loss shows that the price increase from supply tightening is small relative to the missing future revenue, so aggressive policies that artificially restrict supply to raise current prices are unambiguously revenue-reducing in this framework.
	
	\subsection{Effective Competition and Seller Revenue (Cars and Bids)}
	
	Online auction platforms typically report headline participation numbers based on registered bidders, page views, or watchlist additions.  These metrics overstate the competitive intensity that actually determines transaction prices, because many participants are casual browsers who never bid seriously, and the platform's revenue is driven by the small subset of bidders who push prices toward valuations as we documented in the previous section. A natural follow-up question is the policy-relevant magnitude of this gap: what would seller revenue look like if the platform attracted, say, twice as many serious bidders? This question is operationally meaningful because platforms have direct levers for converting casual browsers into active bidders, with qualifying deposits, late-bidding reminders, anti-sniping extensions, and curated notifications being prominent examples. To answer the question, in the Cars and Bids application, we sweep a counterfactual number of serious bidders $N' \in \{5, 8, 10, 14, 20, 30, 50\}$ and compute bounds on seller revenue $\phi \cdot E[\theta^{2:N'}]$ and on winner rent $\phi \cdot E[\theta^{1:N'} - \theta^{2:N'}]$, holding the within-cluster valuation distribution within the estimated bounds.  The first quantity uses the order-statistic bounds in \eqref{eq:E_order_stat_bounds} with $j = 2$; the second uses the closed-form bound on the rent integrand $N'F(v)^{N'-1}(1-F(v))$ over $F(v) \in [F_L(v),F_U(v)]$.  The maintained assumption is that the additional bidders draw from the same valuation distribution as the existing serious bidders.
	
	Figure~\ref{fig:cf_serious_bidders} displays the resulting bounds.  Moving from $N' = 8$ (the main effective-competition estimate) to $N' = 20$ raises the lower seller-revenue envelope by 40--46\% across clusters and the upper envelope by 36--65\%, while winner rents move in the opposite direction: in Cluster~1, for example, the rent interval falls from $[0.23,0.68]$ at $N'=8$ to $[0.10,0.53]$ at $N'=20$, a roughly 25--60\% compression.  These two patterns are economically symmetric: thicker competition redistributes surplus from winners to sellers without changing the value of the underlying allocation, which is the textbook prediction of standard IPV theory and which our bounds confirm in magnitude. The magnitudes provide a price tag for engagement mechanisms: the lower-bound revenue increase by raising effective competition from 8 to 20 bidders, 40--46\% per cluster, can be compared directly with the cost of any platform investment that achieves that increase, such as marketing spend, deposit requirements, or anti-sniping rule changes. An investment whose cost is below this lower-bound revenue gain would be profitable in this partial-equilibrium calculation across all valuation distributions consistent with the data. The rent compression matters for incentive design, since as competition thickens, winner rents fall and the platform's ability to extract surplus through reserve prices, fees, or membership tiers shifts accordingly. A platform with a fee structure calibrated to a low-competition regime may leave revenue on the table as competition grows.
	
	\begin{figure}[p!]
		\centering
		\includegraphics[width=\textwidth]{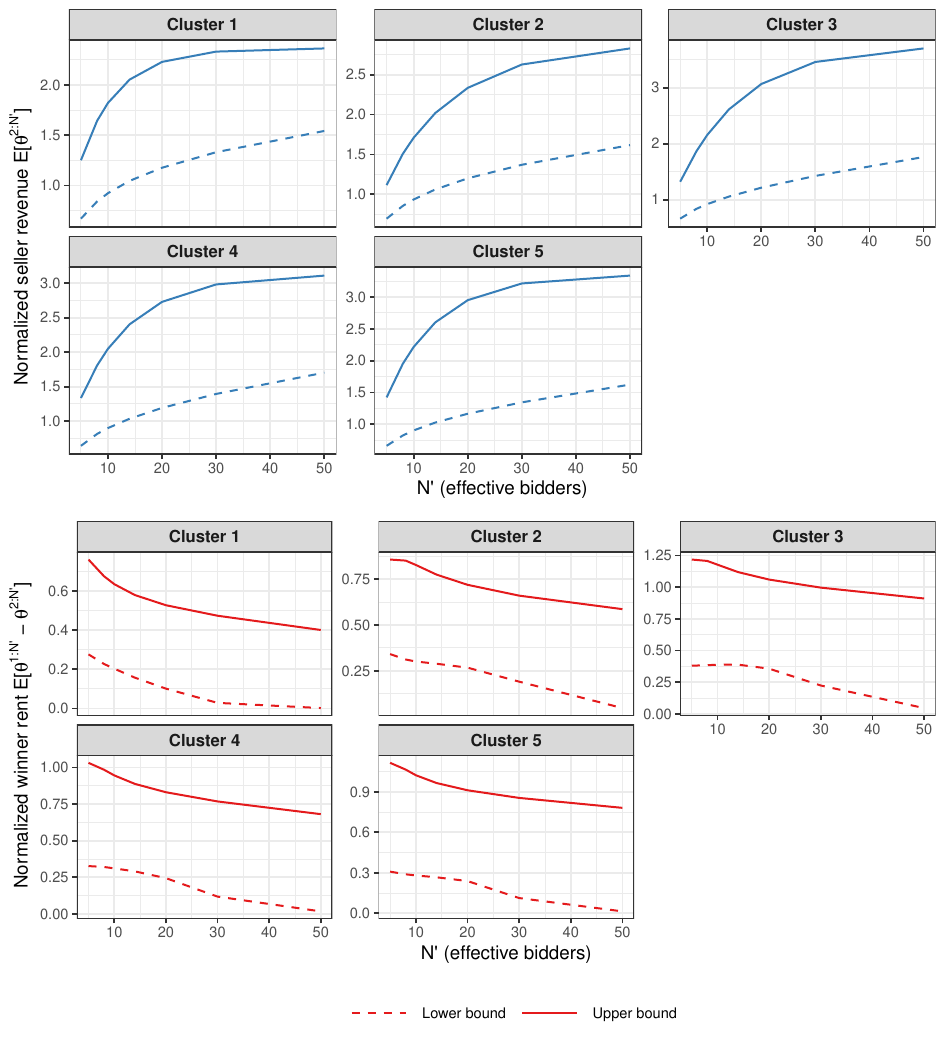}
		\caption{Counterfactual revenue and winner rent as the number of serious bidders~$N'$ varies (Cars and Bids).  Top row: bounds on expected seller revenue $\phi \cdot E[\theta^{2:N'}]$.  Bottom row: bounds on winner rent $\phi \cdot E[\theta^{1:N'} - \theta^{2:N'}]$.  Values are normalized by~$\phi$.  Dashed lines are lower bounds and solid lines are upper bounds.}\label{fig:cf_serious_bidders}
	\end{figure}
	
	\subsection{Reserve Price Envelope (Cars and Bids)}
	
	Reserve pricing is the most direct policy lever available to a seller in an English auction since the reserve sets a floor on the realized transaction price. Classical optimal-auction theory delivers a sharp prescription: raise the reserve until the marginal cost of foregone trade equals the marginal revenue from the reserve payment. However, this prescription requires committing to a particular valuation distribution~$F$.  In our partial identification framework, there is no unique optimal reserve identified.  What replaces the point optimum is a \emph{set} of reserves and an envelope of revenue functions, one for each~$F$ within the bounds. Two natural objects of interest are the lower envelope, which describes the worst-case revenue at each reserve, and the maximin reserve, which maximizes the worst case. Both are interpretable without further assumptions about which~$F$ in $[F_L,F_U]$ is correct, and both are directly actionable by a seller who wishes to hedge against misspecification.
	
	Let $\rho = r/\phi$ denote the normalized reserve.  Under a static second-price auction with $N$ i.i.d.\ bidders drawing from~$F$, expected normalized revenue is
	\begin{equation}\label{eq:reserve_revenue}
		\text{Rev}(\rho; F, N) = \rho \, N(1-F(\rho))F(\rho)^{N-1} + \int_{\rho}^{\bar{\theta}} [1 - I_{F(v)}(N\!-\!1, 2)]\, dv,
	\end{equation}
	where the first term is the expected reserve payment in the event that exactly one bidder clears~$\rho$ (so the price equals the reserve) and the integral is the expected second-highest valuation above~$\rho$ (which determines the price when at least two bidders clear).  The integral is monotone decreasing in~$F$ and is bounded by evaluating it at $F_U$ and~$F_L$.  The reserve-payment term is non-monotone in~$F$ pointwise (it is unimodal with maximum at $F^* = (N-1)/N$) so its bounds are obtained by minimizing and maximizing the integrand over $F(\rho) \in [F_L(\rho), F_U(\rho)]$ at each~$\rho$.  Combining the two terms gives pointwise revenue bounds on $\text{Rev}(\rho; F, N)$ for each~$\rho$.  We sweep $\rho \in [0.5, 2.5]$ at the effective-competition estimate $N = 8$. 
	
	Figure~\ref{fig:cf_reserve} displays the revenue envelope for each cluster.  The lower envelope is bimodal in every cluster: it attains a boundary value at $\rho = 0.5$ that is driven almost entirely by the integral term, drops as $\rho$ rises and the integral term collapses faster than the reserve-payment term builds up, and recovers to an interior local maximum near $\rho \approx 1.3$--$1.4$ where the reserve-payment term contributes more than 95\% of expected revenue.  Which of these two local maxima dominates varies across clusters.  In Clusters~2--4 the interior maximum at $\rho^* = 1.3$, $1.3$, and~$1.4$ exceeds the boundary value by 14.6\%, 23.0\%, and 15.6\% respectively, so the maximin reserve is interior and meaningful.  In Clusters~1 and~5 the two local maxima are essentially tied (the boundary wins by 0.6\% and 4\%, respectively) so the maximin reserve formally sits at $\rho^* = 0.5$, but a near-identical worst-case revenue is achievable at the interior local maximum near $\rho \approx 1.3$--$1.4$.  The upper envelope, by contrast, is maximized at $\rho = 0.5$ in every cluster: under the most favorable distributions within the bounds, $F_U$ rises rapidly above $\rho \approx 1$ and the integral term dominates throughout. The lower and upper envelopes prescribe substantively different reserves in three of the five clusters, and even in the two clusters where they nominally agree at $\rho^* = 0.5$ the lower envelope identifies a near-tie interior local maximum.
	
	This exercise provides useful guidance on reserve price policies. First, the divergence between the maximin and the upper-envelope optima in Clusters~2--4 reveals genuine distributional ambiguity: a planner who is confident that $F$ lies near the upper bound prefers a low reserve, while one who hedges against the worst case prefers a moderate reserve of $\rho^* \approx 1.3$--$1.4$.  The choice between these two policy stances is a value judgment about ambiguity aversion, not an inference question that more data can settle.  Second, the cluster-level heterogeneity is itself a finding: vehicle segments do not share a common optimal reserve, and a platform-wide uniform reserve policy can leave revenue gains relative to a cluster-specific policy in three of five clusters.  Even under the most conservative robustness criterion, segment-targeted reserves can raise revenue by 15--23\% relative to the boundary policy that a one-size-fits-all maximin would prescribe.  
	
	\begin{figure}[p!]
		\centering
		\includegraphics[width=\textwidth]{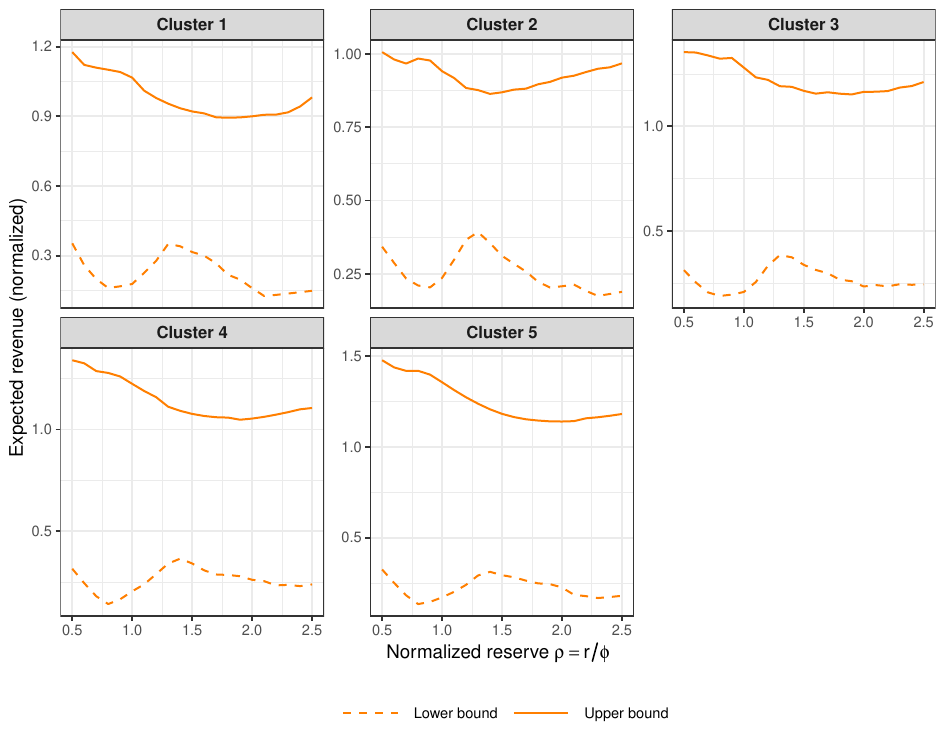}
		\caption{Expected revenue envelope as the normalized reserve~$\rho = r/\phi$ varies (Cars and Bids), evaluated at the effective-competition estimate $N = 8$.  Dashed lines are lower bounds and solid lines are upper bounds.  The maximin reserve $\rho^*$ is the value that maximizes the lower envelope: $\rho^* = 0.5$ in Clusters~1 and~5, and $\rho^* = 1.3$, $1.3$, $1.4$ in Clusters~2,~3,~4 respectively.}\label{fig:cf_reserve}
	\end{figure}

	\section{Conclusion}\label{sec:conclusion}
	
	This paper develops a partial identification approach for the distribution of bidders' valuations in sequential English auctions. Our framework relies on two behavioral restrictions: bidders do not overbid and they account for future auction opportunities, without specifying any particular equilibrium. Our approach has several advantages over existing methods. First, it is robust to a variety of nonstandard features that cause equilibrium-based methods to fail, including supply uncertainty, loss aversion, and non-equilibrium bidding. Second, simulation experiments demonstrate that ignoring the sequential structure leads to misspecified bounds, while our method provides informative bounds that are valid under the maintained assumptions. Third, the method uses non-terminal as well as terminal auction information, yielding tighter bounds than last-period-only approaches when common-value components decline substantially. Two empirical applications demonstrate that the method is practical and yields informative bounds with real data.
	
	From a practical standpoint, the bounds on~$F$ characterize the distribution of buyer valuations without imposing equilibrium assumptions, providing a robust foundation for auction design. For the Cars and Bids platform, the effective-competition analysis suggests that a substantial share of observed bidders are casual participants, so mechanisms to distinguish active from passive bidders may improve price discovery. More broadly, the bounds on~$F$ across market segments serve as inputs to counterfactual policy analysis, including the evaluation of supply uncertainty, effective competition, and reserve price rules. Computing optimal dynamic reserve policies in sequential auctions remains challenging because the interaction between current reserves and future bidder behavior requires equilibrium assumptions that our framework deliberately avoids. Developing such counterfactual analyses within a partial identification framework is an important direction for future work. Other avenues for future research include extending the approach to settings with unobserved heterogeneity and multi-unit demand.

	\bibliographystyle{chicago}
	\bibliography{used_car_paper}
	
	\newpage
	\appendix
	
	\begin{center}
		{\Large\bf Online Supplementary Appendix}
		
		\vspace{0.8em}
		{\large for ``Partial Identification of the Valuation Distribution\\ in Sequential English Auctions''}
		
		\vspace{0.8em}
		Dongwoo Kim, Kyoo il Kim and Pallavi Pal
	\end{center}
	
	\vspace{1.5em}

	\section{Technical Proofs}\label{app:proofs}
	
	\begin{proof}[Proof of Lemma~\ref{lem:lb}]
		Under Assumption~\ref{ass:no_overbid}, $b_{i,k}/\phi(Z_k) \leq \theta_i$ for all $i$ and $k$. Suppose for contradiction that $\eta^{j:n_k} > \theta^{j:n_k}$ for some $j$ and $k$. Then $b^{j:n_k} > \theta^{j:n_k}\phi(Z_k)$, implying that at least $j$ bids exceed the $j$th highest valuation. But this means at least one bidder bid above her valuation, contradicting Assumption~\ref{ass:no_overbid}.
	\end{proof}
	
	\begin{proof}[Proof of Lemma~\ref{lem:ub_last}]
		In the last period $K$, there is no future auction, so $W_K(\theta_i) = 0$. Under Assumption~\ref{ass:opportunity_cost}, for any bidder $i$ with $b_{i,K} < b^{1:n_K}$, the profit from winning at $b^{1:n_K} + \Delta$ must not exceed the opportunity cost, which is zero. Hence $\theta_i\phi(Z_K) - (b^{1:n_K} + \Delta) \leq 0$, giving $\theta_i \leq (b^{1:n_K} + \Delta)/\phi(Z_K)$. For the winner ($b_{i,K} = b^{1:n_K}$), no tighter bound than $\bar{\theta}$ is available.
	\end{proof}
	
	\begin{proof}[Proof of Lemma~\ref{lem:expected_payoff_bound}]
		For notational convenience, let $P_{i,r} \equiv P(b_{i,r} > \max_{j \neq i} b_{j,r})$. We have
		\begin{align*}
			W_k(\theta_i) &\leq \sum_{t=k+1}^{K} \left(\prod_{r=k+1}^{t-1}(1 - P_{i,r})\right) P_{i,t}\, \theta_i\phi(Z_t),
		\end{align*}
		using $\mathbf{E}[\pi_{i,t} \mid \text{win}] \leq \theta_i\phi(Z_t)$ (the maximum profit is from paying zero) and $\tau_{i,t} \leq 1$. We evaluate the sum by examining the last two terms:
		\begin{align*}
			&\left(\prod_{r=k+1}^{K-2}(1-P_{i,r})\right)P_{i,K-1}\,\theta_i\phi(Z_{K-1}) + \left(\prod_{r=k+1}^{K-1}(1-P_{i,r})\right)P_{i,K}\,\theta_i\phi(Z_K).
		\end{align*}
		Using $P_{i,K} \leq 1$:
		\begin{align*}
			&\leq \left(\prod_{r=k+1}^{K-2}(1-P_{i,r})\right)\!\left[P_{i,K-1}\,\theta_i\phi(Z_{K-1}) + (1-P_{i,K-1})\,\theta_i\phi(Z_K)\right] \\
			&= \left(\prod_{r=k+1}^{K-2}(1-P_{i,r})\right)\!\left[P_{i,K-1}\,\theta_i(\phi(Z_{K-1}) - \phi(Z_K)) + \theta_i\phi(Z_K)\right] \\
			&\leq \left(\prod_{r=k+1}^{K-2}(1-P_{i,r})\right)\theta_i\phi(Z_{K-1}),
		\end{align*}
		where the last step uses $P_{i,K-1} \leq 1$ and $\phi(Z_{K-1}) \geq \phi(Z_K)$ (Condition~\ref{cond:descending}). Applying this argument recursively from period $K-1$ back to period $k+1$ yields Part~(a): $W_k(\theta_i) \leq \theta_i\phi(Z_{k+1})$.
		
			For Part~(b), write $V_{i,t}=\theta_i\phi(Z_t)-(p_t+\Delta)$ and $V_{i,K+1}=0$. The condition $V_{i,t}\geq \max\{0,V_{i,t+1}\}$ for $t=k+1,\ldots,K$ implies $V_{i,k+1}\geq V_{i,k+2}\geq\cdots\geq V_{i,K}\geq0$. Since winning period~$t$ requires paying at least $p_t+\Delta$, $\mathbf{E}[\pi_{i,t}\mid \text{win}]\leq V_{i,t}\leq V_{i,k+1}$ for all $t\geq k+1$. Therefore
			\[
			W_k(\theta_i)
			\leq V_{i,k+1}\sum_{t=k+1}^{K}\tau_{i,t}
			\left(\prod_{r=k+1}^{t-1}(1-P_{i,r})\right)P_{i,t}
			\leq V_{i,k+1},
			\]
			because $\tau_{i,t}\leq1$ and the final sum is bounded by the probability of winning in some future period. This proves $W_k(\theta_i)\leq V_{i,k+1}=\theta_i\phi(Z_{k+1})-(p_{k+1}+\Delta)$.
		\end{proof}
	
	\begin{proof}[Proof of Theorem~\ref{thm:ub_general}]
		Applying Lemma~\ref{lem:expected_payoff_bound}(a) to inequality \eqref{eq:opportunity_cost_ineq}:
		\begin{align*}
			\theta_i\phi(Z_k) - (b^{1:n_k} + \Delta) &\leq \theta_i\phi(Z_{k+1}) \\
			\theta_i(\phi(Z_k) - \phi(Z_{k+1})) &\leq b^{1:n_k} + \Delta.
		\end{align*}
		Dividing by $\phi(Z_k) - \phi(Z_{k+1}) > 0$ gives \eqref{eq:ub_general_gross}. If Lemma~\ref{lem:expected_payoff_bound}(b) applies, then
		\begin{align*}
			\theta_i\phi(Z_k) - (b^{1:n_k} + \Delta) &\leq \theta_i\phi(Z_{k+1}) - (p_{k+1} + \Delta) \\
			\theta_i(\phi(Z_k) - \phi(Z_{k+1})) &\leq b^{1:n_k} - p_{k+1}.
		\end{align*}
			Dividing again by the positive denominator gives \eqref{eq:ub_general}. The sharper net statistic is used only for retained observations satisfying $b^{1:n_k}>p_{k+1}$, so that the numerator is positive; otherwise the observation does not contribute to this specification. For a winner with $b_{i,k}=b^{1:n_k}$, the losing-bidder inequalities do not apply, so the only available bound is support boundedness, $\theta_i\leq\bar\theta$.
		\end{proof}
	
		\begin{proof}[Proof of Lemma~\ref{lem:ub_infinite}]
			Let $P_{i,t}(n)\equiv P(b_{i,t}>\max_{j\neq i}b_{j,t}\mid n_t=n)$. We prove the baseline claim by backward induction from bidder $i$'s deadline $d_i$. Since $W_{d_i}(\theta_i,d_i)=0$, the claim is immediate at the deadline. Suppose $W_t(\theta_i,d_i)\leq \theta_i\phi(Z_{t+1})$. Under Condition~\ref{cond:descending}, $W_t(\theta_i,d_i)\leq \theta_i\phi(Z_t)$. Then \eqref{eq:W_infinite} implies
			\begin{align*}
				W_{t-1}(\theta_i,d_i)
				&\leq \int \left\{P_{i,t}(n)\theta_i\phi(Z_t)
				+[1-P_{i,t}(n)]W_t(\theta_i,d_i)\right\}\,dF_N(n) \\
				&\leq \theta_i\phi(Z_t).
			\end{align*}
			Iterating back to $t=k+1$ gives $W_k(\theta_i,d_i)\leq\theta_i\phi(Z_{k+1})$.
			
			The sharper net refinement follows by the same induction with $V_{i,t}=\theta_i\phi(Z_t)-(p_t+\Delta)$ replacing the gross value bound. If $V_{i,t}\geq\max\{0,V_{i,t+1}\}$ along the relevant future path, then the integrand in the display above is bounded by $V_{i,t}$, yielding $W_{t-1}(\theta_i,d_i)\leq V_{i,t}$ and therefore $W_k(\theta_i,d_i)\leq V_{i,k+1}$. The upper bound on $\theta_i$ then follows by the same argument as in the proof of Theorem~\ref{thm:ub_general}.
		\end{proof}
        
    	\begin{proof}[Proof of Lemma~\ref{lem:sharp}]
		This follows directly from Theorem~1 of \cite{chesher2017generalized} applied to the structural function $h(B, Z, U)$ defined in \eqref{eq:structural}. The identified set is characterized by the Artstein inequality applied to the random set $\mathcal{U}(B, Z; h)$ conditional on $Z = z$: for every relevant closed set $S$, the probability that the random set is contained in $S$ cannot exceed the probability that the latent variable $U$ lies in $S$.
    	\end{proof}

	\begin{proof}[Proof of Proposition~\ref{prop:moment_cond}]
		\emph{Part (i) (Upper bound).}  By Lemma~\ref{lem:lb}, $\eta^{j:N_a} \leq \theta^{j:N_a}$ for each auction $a$, so $\Pr(\eta^{j:N_a} \leq v \mid N_a) \geq \Pr(\theta^{j:N_a} \leq v \mid N_a) = I_{F^0(v)}(N_a - j + 1, j)$ for each realization of $N_a$.  Taking expectations over the marginal distribution of $N_a$ yields $\bar G_U(v) \geq \bar I(F^0(v))$, which is \eqref{eq:moment_population}.  The function $\bar I(F)$ is a mixture of regularized incomplete beta functions, each strictly increasing on $(0,1)$ with positive derivative, so $\bar I$ is itself strictly increasing.  Hence the equation $\bar I(F) = \bar G_U(v)$ has a unique solution $F_U^*(v)$, and the inequality $\bar G_U(v) \geq \bar I(F^0(v))$ together with monotonicity gives $F_U^*(v) \geq F^0(v)$.
		
		For consistency, let $\hat G_U(v) = T^{-1}\sum_a \mathbf{1}\{\eta_a^{j:N_a} \leq v\}$.  Under i.i.d.\ sampling and bounded $N_a$, the Glivenko--Cantelli theorem gives $\sup_v |\hat G_U(v) - \bar G_U(v)| = O_p(T^{-1/2})$.  The estimator $\hat F_U(v)$ solves the empirical analog of $\bar I(F) = \bar G_U(v)$.  By the implicit function theorem, at any~$v$ with $\bar I'(F_U^*(v)) > 0$,
		\[
		\hat F_U(v) - F_U^*(v)
		= \frac{T^{-1}\sum_a\left[\mathbf{1}\{\eta_a^{j:N_a}\leq v\}-I_{F_U^*(v)}(N_a-j+1,j)\right]}{\bar I'(F_U^*(v))} + o_p(T^{-1/2}),
		\]
		where $\bar I'(F) = \mathbf{E}[f_{\text{Beta}}(F; N_a - j + 1, j)]$.  Uniform consistency over compact subsets of the support follows from a standard uniform implicit-function argument; at support boundaries where $\bar I'(F_U^*) \to 0$, the convergence rate slows (the boundary problem of Section~\ref{subsec:finite_sample}).  Applying the CLT to the centered moment indicator $\mathbf{1}\{\eta^{j:N_a}\leq v\}-I_{F_U^*(v)}(N_a-j+1,j)$ and then applying the delta method yields the asymptotic distribution \eqref{eq:asymptotic_dist}.  The beta term appears because, under random $N_a$, the empirical beta weights are random as well.  When $N_a$ is treated as fixed, this additional randomness vanishes, leaving $\bar G_U(v)(1 - \bar G_U(v))$ in the numerator and $f_{\text{Beta}}(F_U^*(v); N - j + 1, j)^2$ in the denominator, which reduce to $\hat\sigma(v)^2$ in \eqref{eq:analytical_se}.
		
		\emph{Part (ii) (Lower bound).}  The argument carries over to the lower bound \emph{mutatis mutandis}.  By Theorem~\ref{thm:ub_general}, $h^{g,a}_{1:N_a} \geq \theta^{2:N_a}$ pairwise, so $\Pr(h^g_{1:N_a} \leq v \mid N_a) \leq \Pr(\theta^{2:N_a} \leq v \mid N_a) = I_{F^0(v)}(N_a - 1, 2)$.  Taking expectations over $N_a$ gives $\bar G_L(v) \leq \bar I(F^0(v))$ with shape $(N_a\!-\!1, 2)$, which is the population analog of \eqref{eq:moment_condition_LB}; strict monotonicity of $\bar I$ then yields $F_L^*(v) \leq F^0(v)$, establishing validity.  Glivenko--Cantelli applied to $\hat G_L(v) = T^{-1} \sum_a \mathbf{1}\{h^{g,a}_{1:N_a} \leq v\}$ and the implicit function theorem applied to $\bar I(F) = \bar G_L(v)$ give the influence-function representation
		\[
		\hat F_L(v) - F_L^*(v)
		= \frac{T^{-1}\sum_a\left[\mathbf{1}\{h^{g,a}_{1:N_a}\leq v\}-I_{F_L^*(v)}(N_a-1,2)\right]}{\bar I'(F_L^*(v))} + o_p(T^{-1/2}),
		\]
		and the CLT applied to the lower-bound moment indicator together with the delta method delivers the stated asymptotic normality.  The joint asymptotic distribution is bivariate normal with covariance equal to the covariance of the upper- and lower-bound influence functions.  This covariance is zero for disjoint samples and generally nonzero when the same auction or market can contribute to both endpoints.
	\end{proof}

	\section{Special Case: No Bid Periods}\label{app:special_case}
	
	In practice, some auctions receive no bids when all bidders' valuations fall below the reserve price. We handle this using the following observation.
	
	\begin{observation}\label{obs:nobid}
		If the value of the item in period $k$ is below the reserve price for all bidders, then bidders treat the expected profit from period $k$ as zero when calculating future expected payoffs.
	\end{observation}
	
	\noindent This follows from Assumptions~\ref{ass:no_overbid}--\ref{ass:opportunity_cost}: if no bidder's valuation exceeds the reserve price, no bidder participates, and the anticipated profit from that period is zero.
	
	\begin{lemma}\label{lem:nobid}
		Suppose no bidder bids in an interior period $k$ and the sequence obtained by deleting period~$k$ satisfies the descending condition at period $k-1$, so that $\phi(Z_{k-1})>\phi(Z_{k+1})$. Then for all $i \in I_{k-1}$ with $b_{i,k-1} < b^{1:n_{k-1}}$:
		\begin{equation}\label{eq:nobid}
			\theta_i \leq \frac{b^{1:n_{k-1}} + \Delta}{\phi(Z_{k-1}) - \phi(Z_{k+1})}.
		\end{equation}
	\end{lemma}
	
	\begin{proof}
		Since all bidders' valuations are below the reserve price in period $k$, the expected profit from period $k$ is zero. Hence the opportunity cost in period $k-1$ skips period $k$ and depends on the payoff from period $k+1$ onward. Applying the same baseline argument as in Theorem~\ref{thm:ub_general}, with $\phi(Z_{k+1})$ replacing the next available common-value component, gives
		\[
		\theta_i\{\phi(Z_{k-1})-\phi(Z_{k+1})\}\leq b^{1:n_{k-1}}+\Delta.
		\]
		The denominator is positive by the maintained condition $\phi(Z_{k-1})>\phi(Z_{k+1})$, so division gives the stated bound.
	\end{proof}
	
	\noindent Under Observation~\ref{obs:nobid}, removing no-bid auctions from the data and applying Theorem~\ref{thm:ub_general} to the remaining consecutive auctions is equivalent for the baseline bound whenever the deleted sequence still satisfies the relevant descending filter.
	
	\section{Robustness: Continuation Value in Korean Auctions}\label{app:korean_monotonicity}
	The baseline lower bound in the Korean market application uses the adjacent-opportunity specification $C_k=\hat\phi(Z_{k+1})$ and $m_k=0$ from Section~\ref{subsec:continuation_filter}. This specification treats the next car in the sequence as the relevant continuation opportunity and permits all adjacent descending pairs.  This appendix reports two more conservative implementations.  The first sets
	\[
	C_k^{\max}=\max_{t\geq k+1}\hat\phi(Z_t),
	\]
	so an auction contributes only if its common value exceeds every future common value in the same auction-day category.  The second keeps only tail-to-terminal subsequences satisfying Condition~\ref{cond:descending}, $\hat\phi(Z_k)>\hat\phi(Z_{k+1})>\cdots>\hat\phi(Z_K)$.  Both robustness specifications use the baseline statistic and are combined with the terminal-period lower-bound moment. For comparison, the table also repeats the adjacent baseline and sharper net specifications from Figure~\ref{fig:korean_bounds}.
	
	Figure~\ref{fig:korean_continuation_robustness} and Table~\ref{tab:korean_monotonicity} report the resulting bounds.  The adjacent specification is the most informative in the categories where the sequential moment adds content. The loss from the conservative full-future and tail-to-terminal specifications is concentrated in the left tail, where the additional filters delay the rise of the lower bound; outside that region, the bounds are largely unchanged.
	
	\begin{figure}[p]
		\centering
		\includegraphics[width=0.94\textwidth]{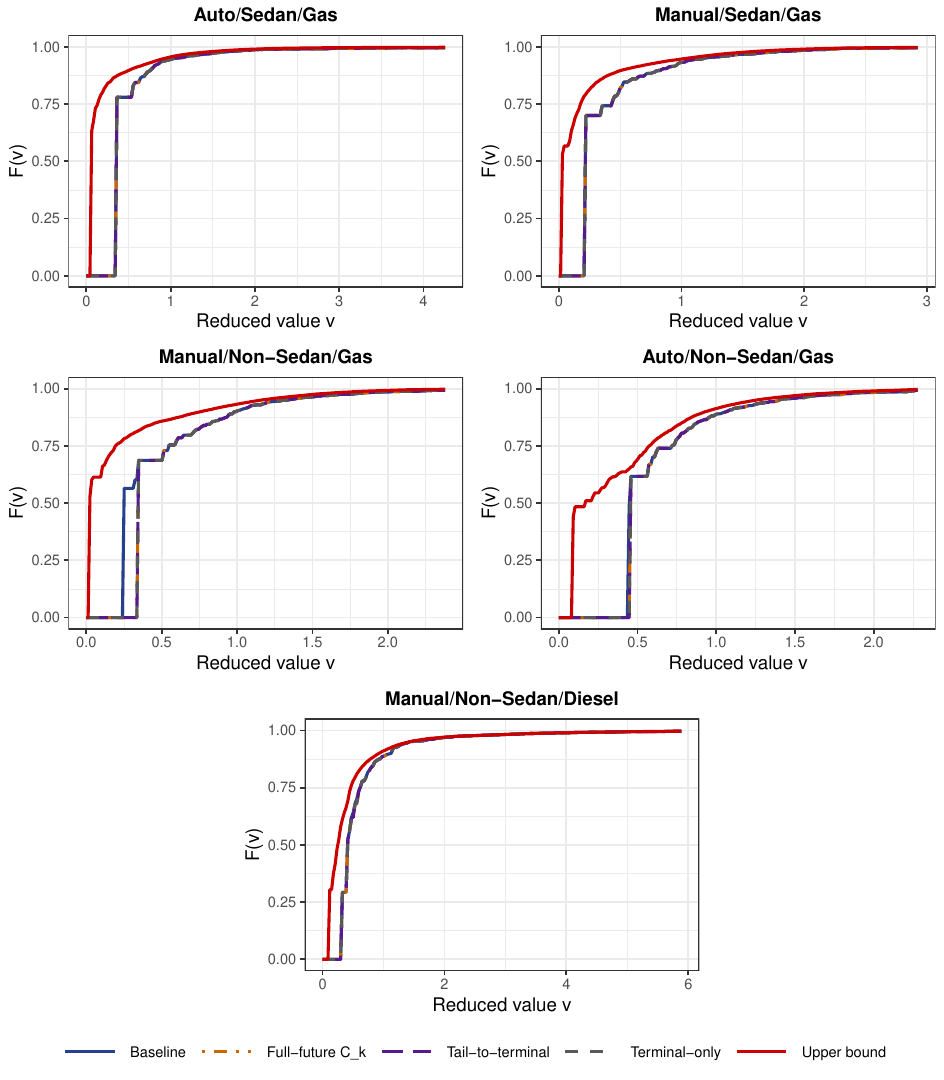}
		\caption{Robustness of bounds to the continuation-value specification in the Korean market.  Red solid: no-overbidding upper bound.  Blue solid: baseline lower bound with $C_k=\hat\phi(Z_{k+1})$.  Orange dot-dashed: full-future lower bound with $C_k=\max_{t\geq k+1}\hat\phi(Z_t)$.  Purple long-dashed: tail-to-terminal lower bound using decreasing subsequences through the terminal auction.  Gray dashed: terminal-only lower bound.}\label{fig:korean_continuation_robustness}
	\end{figure}
	
	\begin{table}[h!]
		\centering
			\caption{Korean lower-bound robustness and adjacent-specification benchmarks.}\label{tab:korean_monotonicity}
		\resizebox{\textwidth}{!}{%
			\begin{tabular}{lrrrrrrrrr} \toprule
				& \multicolumn{2}{c}{Baseline} & \multicolumn{2}{c}{Sharper net} & \multicolumn{2}{c}{$C_k^{\max}$} & \multicolumn{2}{c}{Tail-to-terminal} & Terminal \\
				\cmidrule(lr){2-3}\cmidrule(lr){4-5}\cmidrule(lr){6-7}\cmidrule(lr){8-9}
				Category & Pairs & Width & Pairs & Width & Pairs & Width & Pairs & Width & Width \\ \midrule
				Auto/Sedan/Gas            & 5{,}070 & 0.067 & 3{,}682 & 0.050 & 481 & 0.067 & 55 & 0.067 & 0.067 \\
				Manual/Sedan/Gas          & 3{,}812 & 0.065 & 3{,}165 & 0.065 & 505 & 0.065 & 71 & 0.065 & 0.065 \\
				Manual/Non-Sedan/Gas      & 2{,}942 & 0.106 & 2{,}372 & 0.111 & 509 & 0.129 & 64 & 0.129 & 0.129 \\
				Auto/Non-Sedan/Gas        & 2{,}076 & 0.107 & 1{,}567 & 0.110 & 495 & 0.110 & 50 & 0.110 & 0.110 \\
				Manual/Non-Sedan/Diesel   & 1{,}917 & 0.029 & 1{,}149 & 0.029 & 380 & 0.029 & 59 & 0.029 & 0.029 \\ \bottomrule
			\end{tabular}%
		}
		\vspace{0.6ex}
		\begin{minipage}{\textwidth}
			\footnotesize \emph{Note:} Width is the mean bound width over the evaluation grid after clipping finite-sample crossings at zero width. The baseline and sharper net columns correspond to the adjacent sequential lower-bound specifications plotted in Figure~\ref{fig:korean_bounds}, each combined with the terminal-period moment.  The $C_k^{\max}$ and tail-to-terminal columns are conservative robustness implementations of Section~\ref{subsec:continuation_filter}.  The no-overbidding upper bound is unchanged across columns.
		\end{minipage}
	\end{table}
	
	The conservative filters substantially reduce the number of usable sequential observations: from 15{,}817 baseline pairs to 2{,}370 pairs under $C_k^{\max}$ and 299 pairs under the tail-to-terminal filter.  As a result, the robustness bounds are weakly less informative than the baseline in the categories where the sequential moment is useful, mainly because they rise later in the left tail.  This pattern is expected: replacing $C_k=\hat\phi(Z_{k+1})$ with a full-future continuation-value index protects against later higher-value opportunities, but it does so by discarding most within-day variation.  The robustness exercise therefore clarifies the empirical content of the Korean application: the tight baseline bounds come from the adjacent-opportunity interpretation, while the terminal-only, conservative full-future, and tail-to-terminal specifications provide more robust but wider (conservative) benchmarks.
	
	\section{Effective Competition in Cars and Bids}\label{app:effective_N}
	
	The Cars and Bids auctions have about 14~unique bidders per auction (range 2--44, depending on the filtered sample), each placing approximately two bids on average.  Our baseline bounds use the exact number of unique bidders~$N_a$ per auction in the moment-condition inversion.  However, many participants are casual browsers who place low early bids and never seriously compete for the item.  The IPV assumption requires that all~$N$ bidders draw from the same valuation distribution~$F$, but if casual bidders have systematically lower valuations, the effective competition may be lower than~$N_a$.  This section estimates the ``effective $N$'' using four complementary diagnostics.
	
	The first method uses a competitive bid threshold.  For each auction, count bidders whose maximum bid exceeds a fraction~$\alpha$ of the winning bid.  At $\alpha = 50\%$ (bidders who bid at least half the final price), the median effective~$N$ is~8 across all clusters; at $\alpha = 75\%$, it drops to~5.  These estimates are remarkably stable across vehicle segments. The second method counts active final-stage bidders, those who place bids in the final days of the 7-day listing and are actively competing to win.  Due to timestamp coarseness in our data (recorded in weeks/months rather than hours), this method yields reliable estimates for only a small fraction of auctions.  Where available, the median is approximately~5 bidders in the final 2~days. The third method exploits order statistic spacing.  Under IPV with $N$ i.i.d.\ draws, the gaps between consecutive order statistics follow a known distribution.  We estimate effective~$N$ from the ratio of the price range to the mean bid spacing among the top half of bidders.  This yields a median of 6--7 across all clusters. The fourth method uses a bid concentration ratio.  Count the number of bidders whose maximum bids fall within the top 90\% of the price range (from the lowest to the highest bid).  Bidders below this threshold contribute negligibly to price formation.  The median is 10--11 across clusters. Table~\ref{tab:effective_N} summarizes these estimates by cluster.
	
	\begin{table}[h!]
		\centering
		\caption{Estimates of effective competition in Cars and Bids auctions.}\label{tab:effective_N}
		\begin{tabular}{l r r r r r r} \toprule
			Method & \multicolumn{5}{c}{Median $N_{\text{eff}}$ by cluster} & Overall \\
			\cmidrule(lr){2-6}
			& Cl.~1 & Cl.~2 & Cl.~3 & Cl.~4 & Cl.~5 & \\ \midrule
				Unique bidders (baseline, approx.) & 14 & 14 & 14 & 14 & 14 & 14 \\
			\addlinespace
			M1: Bid $> 50\%$ of winner      & 9  & 8  & 8  & 8  & 8  & 8 \\
			M1: Bid $> 75\%$ of winner      & 5  & 5  & 5  & 5  & 4  & 5 \\
			M2: Final 2 days                 & \multicolumn{5}{c}{$\approx 5$ (sparse data)} & 5 \\
			M3: Order statistic spacing      & 6  & 7  & 7  & 7  & 6  & 7 \\
			M4: 90\% concentration           & 10 & 10 & 11 & 11 & 10 & 10 \\
			\bottomrule
		\end{tabular}
		\vspace{0.6ex}
		\begin{minipage}{\textwidth}
			\footnotesize \emph{Note:} M1 counts bidders whose maximum bid exceeds $\alpha\%$ of the winning bid.  M2 counts bidders active in the final 2~days (limited by timestamp granularity).  M3 estimates $N$ from the ratio of the top-half price range to mean bid spacing.  M4 counts bidders within the top 90\% of the full bid range.
		\end{minipage}
	\end{table}
	
	The four methods broadly point to an effective competition level of $N_{\text{eff}} \approx 5$--$10$, substantially below the roughly 14~unique bidders in a typical auction.  Methods~1 (at 50\%) and~3 (spacing) agree closely at $N_{\text{eff}} \approx 7$--$8$, suggesting that roughly half of the unique bidders are casual participants whose bids never approach the transaction price.
	
	This finding has two implications for bounds estimation.  First, our baseline results use the actual~$N_a$, which may moderately overstate the effective competition.  Second, re-estimating the bounds with smaller~$N_{\text{eff}}$ produces wider but more conservative bounds.  Figure~\ref{fig:cb_neff_bounds} displays the panel upper bound alongside the sequential lower bound computed at representative effective-competition values from the plotted diagnostics.  Methods~1 and~3 nearly eliminate crossing with the panel upper bound, while Method~4 reduces but does not eliminate crossing in some clusters. This suggests that the baseline bounds with the actual-$N_a$ are more aggressive than what the panel diagnostic supports, and removing the casual or low-intensity bidders restores non-crossing bounds. The qualitative conclusions are robust across the range $N_{\text{eff}} \in [6, 14]$.
	
	\begin{figure}[h!]
		\centering
		\includegraphics[width=\textwidth]{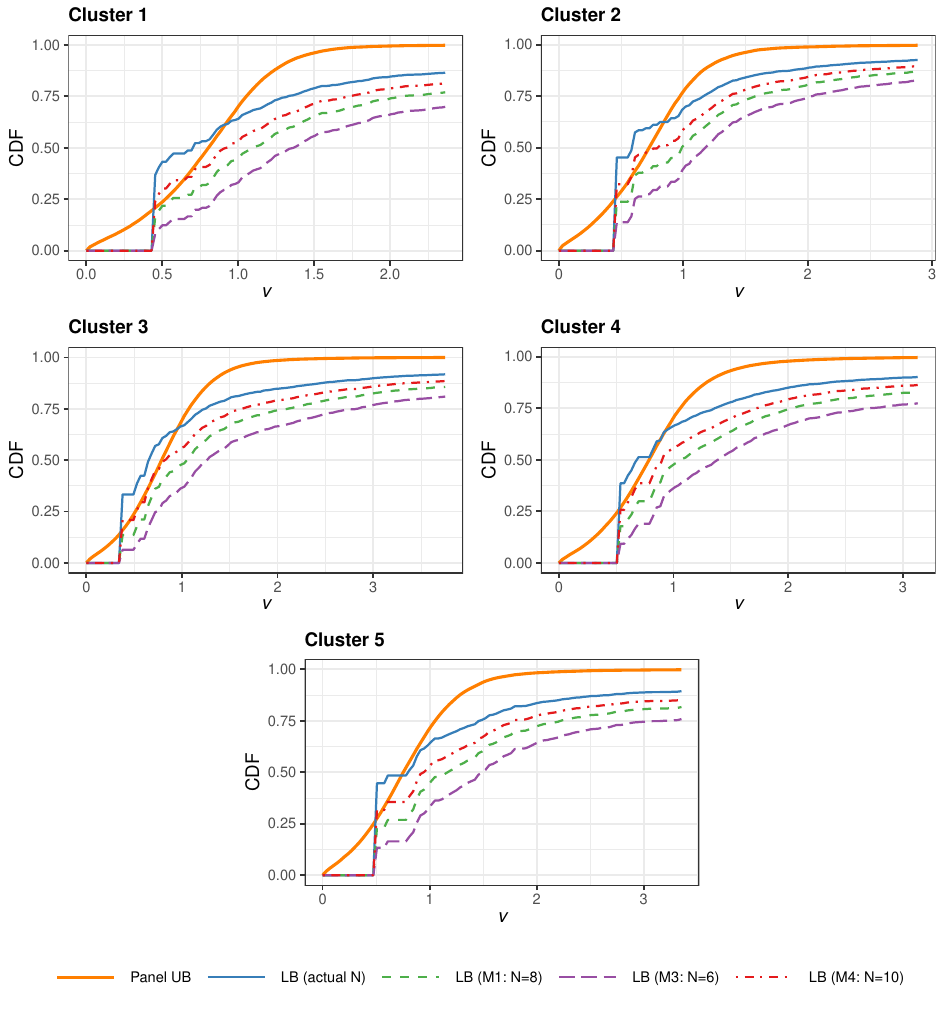}
		\caption{Sensitivity of Cars and Bids bounds to effective competition.  Each panel shows the panel upper bound on~$F$ (orange solid, computed from individual bidders' maximum reduced bids without beta inversion) and the sequential lower bound on~$F$ at representative effective~$N$ levels: baseline $N_a$ (blue solid), M1 at 50\% threshold with $N_{\text{eff}} = 8$ (green dashed), M3 spacing with $N_{\text{eff}} = 6$ (purple long-dashed), and M4 concentration with $N_{\text{eff}} = 10$ (red dot-dashed).}\label{fig:cb_neff_bounds}
	\end{figure}
	
\end{document}